\documentclass[11pt]{article}
\usepackage{fullpage}
\usepackage{tikz}
\usepackage{graphicx}
\usepackage{verbatim}
\usepackage{amssymb,amsfonts,amsmath,amsthm}
\usepackage{url}
\usepackage[pdftex,colorlinks,linkcolor=blue,citecolor=blue,filecolor=blue,urlcolor=blue]{hyperref}
\usepackage[capitalise]{cleveref}
\usepackage{algorithm} % Boxes/formatting around algorithms
\usepackage[noend]{algpseudocode} % Algorithms
\usepackage{color}
\usepackage{ifthen}
\usepackage{booktabs}       % professional-quality tables
\usepackage{nicefrac}       % compact symbols for 1/2, etc.
\usepackage{microtype}      % microtypography
\usepackage{mathtools}
\usepackage{xspace}

\usepackage[toc]{multitoc}

\newtheorem{theorem}{Theorem}[section]
\newtheorem{corollary}[theorem]{Corollary}
\newtheorem{proposition}[theorem]{Proposition}
\newtheorem{lemma}[theorem]{Lemma}
\newtheorem{claim}[theorem]{Claim}
\newtheorem{observation}[theorem]{Observation}

\theoremstyle{definition}
\newtheorem{definition}[theorem]{Definition}
\theoremstyle{remark}
\newtheorem{remark}[theorem]{Remark}

\usepackage{enumitem}
\setlist[itemize]{itemsep=0pt}
\usepackage{thm-restate}
\DeclarePairedDelimiter{\ceil}{\lceil}{\rceil}
\DeclarePairedDelimiter{\floor}{\lfloor}{\rfloor}

\newcommand{\R}{\mathbb R}
\renewcommand{\H}{\mathcal H}
\newcommand{\F}{\mathcal F}

\renewcommand{\Pr}{ \mathrm P}
\newcommand{\eps}{ \varepsilon}

\newcommand{\N}{\mathbb{N}}

\newcommand{\cc}[1][n]{\{0,1\}^{#1}}
\newcommand{\zo}{\{0,1\}}
\newcommand{\PMargs}[2]{\mathrm {SM}_{#1,#2}}
\newcommand{\PM}{\PMargs{n}{k}}

\newcommand{\VC}{ \mathrm {VC}}
\newcommand{\DNF}{ \mathrm {DNF}}
\newcommand{\CNF}{ \mathrm {CNF}}
\newcommand{\EQ}{ \mathrm {EQ}}
\newcommand{\GEQ}{ \mathrm {GEQ}}

\newcommand{\AND}{ \mathrm {AND}}
\newcommand{\OR}{ \mathrm {OR}}
\newcommand{\XOR}{ \mathrm {XOR}}
\newcommand{\GT}{ \mathrm {GT}}
\newcommand{\Disj}{ \mathrm {Disj}}
\newcommand{\LTF}{ \mathrm {LTF}}
\newcommand{\ELTF}{ \mathrm {ELTF}}
\newcommand{\EMAJ}{ \mathrm {EMAJ}}
\newcommand{\SYM}{ \mathrm {SYM}}
\newcommand{\C}{ \mathcal {C}}

\newcommand{\newclass}[2]{\newcommand{#1}{{\text{\upshape\sffamily #2}}\xspace}}
\renewcommand{\P}{{\text{\upshape\sffamily P}}\xspace}
\newclass{\NP}{NP}
\newclass{\MA}{MA}
\newclass{\BPP}{BPP}
\newclass{\TC}{TC}
\newclass{\U}{U}
\newcommand{\mydot}{{\kern.1em{\cdot}\kern.1em}}
\newcommand{\UBPP}{\U\mydot\BPP}
\newcommand{\UP}{\U\mydot\P}

\newcommand\polylog{\mathrm{polylog}}

\title{String Matching: Communication, Circuits, and Learning}

\author{
    Alexander Golovnev\thanks{
    Harvard University.
    Email: \texttt{alexgolovnev@gmail.com}.
    Research supported by a Rabin Postdoctoral Fellowship.
    }
     \and	
    Mika G\"{o}\"{o}s\thanks{
    Institute for Advanced Study, Princeton, NJ, USA.
    Email: \texttt{mika@ias.edu}
    }
    \and	
    Daniel Reichman\thanks{
    Department of Computer Science, Princeton University.
    Email: \texttt{daniel.reichman@gmail.com}
    }
    \and
	Igor Shinkar\thanks{
    School of Computing Science, Simon Fraser University.
    Email: \texttt{ishinkar@sfu.ca}
    }
}

\begin{document}

\maketitle

\begin{abstract}
\noindent
\emph{String matching} is the problem of deciding whether a given $n$-bit string contains a given $k$-bit pattern. We study the complexity of this problem in three settings.
\begin{itemize}[itemsep=4pt]
\item {\bf Communication complexity.}~~%
For small $k$, we provide near-optimal upper and lower bounds on the communication complexity of string matching. For large $k$, our bounds leave open an exponential gap; we exhibit some evidence for the existence of a better protocol.

\item {\bf Circuit complexity.}~~%
We present several upper and lower bounds on the size of circuits with  threshold and DeMorgan gates solving the string matching problem. Similarly to the above, our bounds are near-optimal for small $k$.

\item {\bf Learning.}~~%
We consider the problem of learning a hidden pattern of length at most $k$ relative to the classifier that assigns $1$ to every string that contains the pattern. We prove optimal bounds on the VC dimension and sample complexity of this problem.
\end{itemize}
\end{abstract}

\vspace{3mm}
{\small \tableofcontents}

\thispagestyle{empty}
\setcounter{page}{0}

\newpage

%%%%%%%%%%%%%%%%%%%%%%
\section{Introduction}
%%%%%%%%%%%%%%%%%%%%%%
One of the most fundamental and frequently encountered tasks by minds and machines is that of detecting patterns in perceptual inputs. A basic example is the \emph{string matching} problem, where given a string $x\in\{0,1\}^n$ and a pattern $y\in\{0,1\}^k$, $k\leq n$, the goal is to decide whether $x$ contains $y$ as a substring. Formally, denoting by $x{[i,j]}$ the bits of $x$ in the interval $[i,j]\coloneqq\{i,i+1,\ldots,j\}$, we define a Boolean function by
\[
\PM(x,y)\coloneqq 1\qquad\text{iff}\qquad x{[i,i+k-1]}=y\enspace\text{for some $i\in[n-k+1]$}.
\]
String matching is well-studied in the context of traditional algorithms: it can be computed in linear time~\cite{boyer1977fast,knuth1977fast,galil1983time} (with some lower bounds given by~\cite{rivest1977worst}). It has also been studied in more modern algorithmic frameworks such as streaming~\cite{porat2009exact}, sketching~\cite{bar2004sketching}, and property testing~\cite{ben2017deleting}. See \Cref{sec:related} for more related work.

In this work we study the $\PM$ problem in three models of computation, where it appears to have received relatively little attention.
\begin{enumerate}
\item \emph{Communication complexity:}
How many bits of communication are required to compute $\PM$ when the input $(x,y)$ is adversarially split between two players?
\item \emph{Circuit complexity:}
How many gates are needed to compute $\PM$ by DeMorgan circuits (possibly in low depth)? How about threshold circuits?
\item \emph{Learning:}
How many labeled samples of strings must be observed in order to (PAC) learn a classifier assigning $1$ to a string if and only
if it contains a (fixed) hidden pattern $y$? What is the VC dimension of this problem?
\end{enumerate}

%%%%%%%%%%%%%%%%%%%%%%%%%%%%%%%%%%%%
\subsection{Results: Communication Complexity}
%%%%%%%%%%%%%%%%%%%%%%%%%%%%%%%%%%%%

We first show bounds on the randomized two-party communication complexity of $\PM$. (For standard textbooks on communication complexity, see~\cite{kushilevitz97communication,jukna2012boolean}.) The only related prior work we are aware of is Bar-Yossef et al.~\cite{bar2004sketching} who studied the \emph{one-way} communication complexity of string matching; our focus is on \emph{two-way} communication. Our bounds are near-optimal for small $k$, but for large $k \geq \Omega(n)$, we leave open a mysterious exponential gap. Our protocols work regardless of how the input bits $(x,y)$ are bipartitioned between the players, whereas our lower bound is proved relative to some fixed hard partition.

\begin{restatable}[Communication Complexity]{theorem}{comm}
\label{thm:comm}
For the $\PM(x,y)$ problem:
\begin{itemize}
  \item
{\bf Upper bound:} Under any bipartition of the input bits, there is a protocol of cost
\begin{center}
\begin{tabular}{r l l }
\upshape Deterministic: & $O(\log k\cdot n/k)$ & if  $k\leq \sqrt{n}$ ;\\
\upshape Randomized: & $O(\log n \cdot \sqrt{n})$ & if $k \geq \sqrt{n}$.\\
\end{tabular}
\end{center}
  \item
{\bf Lower bound:} For $k\geq 2$ there is a bipartition of the input bits such that every randomized protocol requires $\Omega(\log\log k\cdot n/k)$ bits of communication, even for the fixed pattern $y = 1^k$.
\end{itemize}
\end{restatable}

\medskip
\begin{remark}
Note that the most natural bipartition, where Alice gets $x$ and Bob gets $y$, is easy. Indeed, for such partition there is a randomized $O(\log n)$-bit protocol, where Bob sends to Alice a hash of $y$, and Alice compares it with the hashes of the substrings $x[1,k]$, $x[2,k+1]$,\ldots, $x[n-k+1,n]$. Under this bipartition, by setting $k=n$, one can also recover the usual \emph{equality} problem, which is well-known to have deterministic communication complexity $\Omega(n)$. This explains why nontrivial protocols for large $k$ need randomness.
\end{remark}

\paragraph{A better protocol?}
For simplicity of discussion, consider the case $k=n/2$.
\begin{center}\itshape
What is the randomized communication complexity of $\PMargs{n}{n/2}$?
\end{center}
Our bounds, $\Omega(\log\log n)$ and $O(\log n\mydot\sqrt{n})$, leave open a huge gap. We conjecture that the answer is closer to the lower bound. As formal evidence we show that problems closely related to $\mathrm{SM}_{n,n/2}$ admit efficient ``unambiguous randomized'' (aka $\U\!\mydot\!\BPP$) communication protocols. A classic result~\cite{yannakakis91expressing} says that any ``unambiguous deterministic'' (aka $\U\!\mydot\!\P$) protocol can be efficiently simulated by a deterministic one, that is, $\U\!\mydot\!\P=\P$ in communication complexity. A randomized analogue of this, $\U\!\mydot\!\BPP=\BPP$, turns out to be false as a consequence of the recent breakthrough of Chattopadhyay et al.~\cite{chattopadhyay19log}. One can nevertheless interpret our $\UBPP$ protocols as evidence for the existence of improved randomized protocols.

\paragraph{Techniques.}
Our lower bound in \Cref{thm:comm} requires proving a tight randomized lower bound for composed functions of the form $\OR\circ\GT$ (where $\GT$ is the \emph{greater-than} function), which answers a question of Watson~\cite{watson18communication}. We observe that the lower bound follows by a minor modification of existing \emph{information complexity} techniques~\cite{braverman16discrepancy}. For upper bounds, the role of periods in strings plays a central role (\Cref{sec:periods}). We go on to discuss a natural \emph{period finding} problem, and conjecture that it is easy for randomized protocols. See \Cref{sec:evidence} for details.

%%%%%%%%%%%%%%%%%%%%%%%%%%%%%%%%%
\subsection{Results: Circuit Complexity}
%%%%%%%%%%%%%%%%%%%%%%%%%%%%%%%%%

\paragraph{Threshold circuits.}
A threshold circuit is a circuit whose gates compute \emph{linear threshold functions} (LTFs). Recall that an LTF outputs 1 on an $m$-bit input $x$ if and only if $\sum_{i\in[m]} a_ix_i\geq \theta$ for some fixed coefficient vector $a \in \R^{m}$, and $\theta \in \R$. The study of threshold circuits is often motivated by its connection to neural networks~\cite{hajnal1993threshold,parberry1988parallel,parberry1994circuit,martens2013representational,muroga1971threshold}. The case of \emph{low-depth} threshold circuits is also interesting. In particular, one line of work~\cite{siu1993depth,razborov1992small,siu1991power} has focused on efficient low-depth threshold implementations of arithmetic primitives (addition, comparison, multiplication). As for lower bounds, \cite{hajnal1993threshold} show an exponential-in-$n$ lower bound for the \emph{mod-2 inner-product} function against depth-$2$ threshold circuits of low weight (see \cite{forster2001relations} for an extension). Superlinear lower bounds on the number of gates of arbitrary depth-$2$ as well as low-weight depth-3 threshold circuits were proven recently by Kane and Williams~\cite{kane2016super}.

It is important that we measure the size of a threshold circuit as the \emph{number of gates} (excluding inputs), in which case even superconstant lower bounds are meaningful. For example, it is easy to implement the equality function (namely $\PMargs{n}{n}$) using three threshold gates (albeit, with exponential weights). Thus, in contrast to the case of bounded fanin circuits, proving linear or even nonconstant lower bounds on the number of gates is not straightforward. Indeed, there are few explicit examples of functions with superconstant lower bounds \cite{groeger1993linear}, and proving them is considered challenging~\cite{roychowdhury1994lower}. Indeed, Jukna~\cite{jukna2012boolean} writes ``even proving non-constant lower bounds $\ldots$ is a nontrivial task".

We show that $\PM$ admits a linear-size implementation at low depth. Thereafter we focus on its fine-grained complexity, seeking to establish lower bounds as close to $\Omega(n)$ as possible.

\begin{restatable}[Threshold circuits]{theorem}{threshold}
\label{thm:threshold}
For the $\PM(x,y)$ problem:
\begin{itemize}
  \item
    {\bf Upper bound:} There is a depth-2 threshold circuit of size $O(n - k)$.
%\item
%    {\bf Lower bound for depth 2:} Suppose $2.1\cdot\log n \leq k \leq n$. Then any depth 2 threshold circuit for $\PM$ must be of size at least $\Omega(n/k^2)$.
\item
    {\bf Lower bound for unbounded depth:} Any threshold circuit must be of size
\begin{center}
\begin{tabular}{l l }
\upshape $\Omega(\frac{n\log\log{k}}{k\log{n}})$ & if $k>1$;\\
\upshape $\Omega(\sqrt{n/k})$ & if $k \geq 2.1\cdot\log n$.\\
\end{tabular}
\end{center}
\end{itemize}
\end{restatable}
The second lower bound is stronger than the first one in the regime $k=\Omega(n\cdot(\frac{\log\log{n}}{\log{n}})^2)$. We note that for $k \leq \polylog(n)$, we have nearly linear lower bounds
%(i.e., $\Omega(n/ \poly(\log(n)))$)
for unbounded-depth threshold circuits computing $\PM$.
%If we restrict our attention to depth--$2$ threshold circuits, then for $k = \Theta(\log(n))$ we have a nearly linear (e.g., $\Omega(n/ \log ^2 n)$)
%lower bound for depth--$2$ threshold circuits computing $\PM$.
We stress that there are no restrictions on the weights of the threshold gates in these lower bounds.
%Our proof also implies (weaker) lower bounds for $k<2.1 \log n$ (see \Cref{sec:lb threshold unbdd k = logn} for details).
We are not able to prove $\Omega(n)$
lower bounds even for depth-2 threshold circuits. Proving such lower bounds (or constructing a threshold circuit of size $o(n)$) remains open. We can prove strong lower bounds for depth-$2$ circuits in some special cases (see \Cref{sec:thrdepth2}).
%%Note that our proof of lower bounds works only for the parameters $k \leq n \leq \frac{k}{2} \cdot 2^{k/2}$,
%%that is, roughly for the regime of $k$ between $2\log(n)$ and $n$.

\paragraph{Techniques.} In \Cref{sec:lb threshold unbdd k = logn} we obtain lower bounds for threshold circuits from the lower bounds on communication complexity of $\PM$ using a connection between threshold complexity and circuit complexity outlined by \cite{nisan1993communication}. We also prove lower bounds for threshold circuits by reducing the problem of computing a ``sparse hard'' function to computing $\PM$. Perhaps surprisingly, we show that the string matching problem can encode a truth table of an arbitrary sparse (few preimages of $1$) Boolean function.

\paragraph{DeMorgan circuits.}
We consider usual DeMorgan circuits (AND, OR, NOT gates) of \emph{unbounded fan-in} and show upper and lower bounds on the circuit complexity of $\PM$. We emphasize again that we measure the size of a circuit as the \emph{number of gates} (excluding inputs). For example, the $n$-bit AND can be computed with a circuit of size 1.

We start by analyzing the case of low-depth circuits.

\begin{restatable}[Depth-2 DeMorgan circuits]{theorem}{depthtwolb}
\label{thm:depth2lb}
For the $\PM(x,y)$ problem:
\begin{itemize}
  \item
{\bf Depth-2 upper bound:} There is a depth-2 DeMorgan circuit of size $O(n \cdot 2^{k})$.
  \item
{\bf Depth-2 lower bound:} Any depth-2 DeMorgan circuit must be of size
\begin{center}
\begin{tabular}{ l l l }
 $\Omega(n\cdot2^{k})$ & & if  $1<k\leq\sqrt{n}$ ;\\
 $\Omega(2^{2\sqrt{n-k+1}})$ & & if $k\geq\sqrt{n}$.\\
\end{tabular}
\end{center}
\end{itemize}
\end{restatable}

For $k\leq\sqrt{n}$, our depth-2 results are optimal (up to a constant factor). For large $k$, say $k=n/2$, there is (similarly as for communication) a huge gap in our bounds: $2^{\Omega(\sqrt{n})}$ versus $2^{O(n)}$. We do not know what bound to conjecture here as the correct answer.

For DeMorgan circuits, the celebrated H{\aa}stad's switching lemma \cite{haastad1987computational} established exponential lower bounds for bounded depth circuits computing explicit functions (e.g., majority, parity). We note that in contrast to the parity function, the string matching function admits a polynomial size circuit of depth 3. It is unclear (to us) how to
leverage known tools for proving lower bounds for small depth circuits (such as the switching lemma) towards proving super linear lower bounds for small depth DeMorgan circuits computing $\PM$. Whether the string matching problem can be computed by a depth $3$ (or even unrestricted) DeMorgan circuit of size $O(n)$ remains open.

\bigskip\noindent
Next, we prove that the circuit complexity of $\PM$ for general DeMorgan circuits (unrestricted depth and fan-in) must be $\Omega(n)$. We also include a relatively straightforward upper bound (which may have been discovered before; \cite{galil1985optimal} claims an upper bound $O(n \log^2 n)$ without a proof).

\begin{restatable}[General DeMorgan circuits]{theorem}{elimination}
\label{thm:lb DeMorgan gate elimination}
For the $\PM(x,y)$ problem:
\begin{itemize}
  \item
{\bf Upper bound:} There is a DeMorgan circuit of size $O(n k)$ and depth $3$.
  \item
{\bf Lower bound:} Any DeMorgan circuit must be of size at least $n/2$.%
\end{itemize}
\end{restatable}

\paragraph{Techniques.} We prove the lower bound on $\DNF$ by exhibiting an explicit set of inputs to $\PM$ each of which requires a separate clause in any $\DNF$. Our lower bound for $\CNF$ involves estimating the size of maxterms of $\PM$. For the lower bound against circuits of unrestricted depth, we adjust the gate elimination technique to the case of unbounded fan-in circuits. See Section~\ref{sec:depth2} for details.

%\noindent In fact, a similar argument proves a slightly stronger lower bound of $n$ for $\log{n}\leq k \leq n-\log{n}$. Since in this paper we are focused on asymptotic results for all regimes of $k$, we omit this proof.

%%%%%%%%%%%%%%%%%%%%%%%%%%%%%%%%%%%%%%%%%%%%%%%%%%%%%%%%%%%%%
\subsection{Results: Learning}
%%%%%%%%%%%%%%%%%%%%%%%%%%%%%%%%%%%%%%%%%%%%%%%%%%%%%%%%%%%%%
Finally, we seek to understand the sample complexity of PAC-learning the string matching function $\PMargs{n}{\ell}(x,\sigma)$, where $x$ is an arbitrary string of length $n$
and $\sigma$ is a \emph{fixed} pattern of length $\ell \leq k$. Towards this goal
we prove (almost) tight bounds on the $\VC$ dimension of the class of these functions. The $\VC$ dimension essentially determines the sample complexity needed to learn the pattern $\sigma$ from a set of i.i.d. samples in the PAC learning framework. We formalize these notions below.

Let $\Sigma$ be a fixed finite alphabet of size $|\Sigma|\geq2$.\footnote{In contrast to the circuit and communication setting, for the learning problem we consider nonbinary alphabets.} By $\Sigma^n$ we denote the set of strings over $\Sigma$ of length $n$, and by $\Sigma^{\leq k}$ we denote the set of strings of length at most $k$. We study the $\VC$ dimension of the class of functions, where each function is identified with a pattern of length at most $k$, and outputs $1$ only on the strings containing this pattern. Recall that the length of the pattern $k=k(n)\leq n$ can be a function of $n$. We now define the set of functions we wish to learn:

\begin{definition}
For a fixed \emph{finite} alphabet $\Sigma$ and an integer $k>0$, let us define the class of Boolean functions $\mathcal{H}_{k,\Sigma}$ over $\Sigma^n$ as follows.
Every function $h_{\sigma} \in \mathcal{H}_{k,\Sigma}$ is parameterized by a pattern $\sigma \in \Sigma^{\leq k}$ of length at most $k$. Hence,
$|\mathcal{H}_{k,\Sigma}|=\frac{|\Sigma|^{k+1}-1}{|\Sigma|-1}$. For a string $s\in\Sigma^n$, $h_{\sigma}(s)=1$ if and only if $s$ contains $\sigma$ as a substring.
\end{definition}

To analyze the sample complexity required to learn a function from $\mathcal{H}_{k,\Sigma}$ we first define \emph{VC dimension}.
\begin{definition}
Let $\F$ be a class of functions from a set $D$ to $\{0,1\}$, and let $S\subseteq D$. A \emph{dichotomy} of $S$ is one of the possible labellings of the points of $S$ using a function from $\F$.
%The restriction of $\F$ to $S=\{s_1,\ldots,s_d\}$ is the following class of binary vectors $\F_S=\{(f(s_1), \ldots, f(s_d))\colon f\in\F\}$. 	
$S$ is \emph{shattered} by $\F$ if $\F$ realizes all $2^{|S|}$ dichotomies of $S$.
The \emph{$\VC$ dimension} of $\F$, $\VC(\F)$, is the size of the largest set $S$ shattered by $\F$.
\end{definition}

In particular, $\VC(\mathcal{H}_{k,  \Sigma})=d$ if and  only if there is a set $S$ of $d$ strings of length $n$ such that for every $S' \subseteq S$,
there exists a pattern $P_S$ of length at most $k$ occurring in all the strings in $S'$ and not occurring in all the strings in $S \setminus S'$.

A class of functions $\F$ is PAC-learnable\footnote{For a precise definition of PAC learning, see Definition~\ref{def:PAC}.} with accuracy $\varepsilon$ and confidence $1-\delta$ in $\Theta\left(\frac{\VC(\F)+\log(1/\delta)}{\varepsilon}\right)$ samples~\cite{blumer1989learnability,ehrenfeucht1989general,hanneke2016optimal}, and is agnostic PAC-learnable in $\Theta\left(\frac{\VC(\F)+\log(1/\delta)}{\varepsilon^2}\right)$ samples~\cite{anthony2009neural,shalev2014understanding}. Thus, tight bounds on the $\VC$ dimension of a class of functions give tight bounds on its sample complexity.

Our main result is a tight bound on the $\VC$ dimension of $\mathcal{H}_{k,\Sigma}$ (up to low order terms). That is:

%\begin{theorem}\label{thm:VCmain}
\begin{restatable}{theorem}{VCmain}
\label{thm:VCmain}
Let $\Sigma$ be a finite alphabet of size $|\Sigma|\geq2$, then
$$\VC(\mathcal{H}_{k,  \Sigma}) = \min(\log{|\Sigma|}(k-O(\log k)), \log n+O(\log \log n)) \; .$$
%\end{theorem}
\end{restatable}

It follows that the sample complexity of learning patterns is $O(\log n)$. We also show that there are efficient polynomial time algorithms solving this learning problem.
See Corollary~\ref{cor:runtime} for details.

\paragraph{Techniques.}
We prove our upper bound on the VC dimension by a double counting argument. This argument uses Sperner families to show that shattering implies a ``large'' family of non-overlapping patterns, which, on the other hand, is constrained by the length $n$ of the strings
that we shatter. The lower bound is materialized by the idea to have $2^d$ patterns $P=\{p_0 \ldots p_{2^{d}-1}\}$ and $d$ strings such that the $i$th string is a concatenation of all patterns
with the binary expansion of their index having the $i$th bit equal $1$. We construct a family of patterns $T$ with the property that for any pair of distinct strings $\alpha,\beta \in T$, their concatenation $\alpha \beta$ does not contain a string $\gamma \in T, \gamma \neq \alpha, \beta$. Using this family (with some additional technical requirements) we are able to show that $P$ shatters a set of $d$ strings implying our lower bound on the VC dimension.
%%%%%%%%%%%%%%%%%%%%%%
\section{More related work} \label{sec:related}
%%%%%%%%%%%%%%%%%%%%%%
%%%%%%%%%%%%%%%%%%%%%%%%%%%%%%%%%
\paragraph{Circuit complexity.}
%%%%%%%%%%%%%%%%%%%%%%%%%%%%%%%%%

\emph{Upper bounds} on the circuit complexity of 2D image matching problem under projective transformations was studied in \cite{rosenke2016exact}. In this problem, which is considerably more complicated than the pattern matching problems we study, the goal is to find a projective transformation $f$ such that $f(A)$ ``resembles''\footnote{We refer to \cite{rosenke2016exact} for the precise definition of distance used there.} $B$ for two images $A, B$. Here, images are 2D square arrays of dimension $n$ containing discrete values (colors). In particular, it is proven that this image matching problem is in $\TC^0$ (it admits a threshold circuit of polynomial size and logarithmic depth in $n$). These results concern a different problem than the string matching considered here, and do not seem to imply the upper bounds we obtain for circuits solving the string matching problem.

The idea to lower bound the circuit complexity of Boolean functions that arise in feature detection was studied in \cite{legenstein2001foundations,legenstein2002neural}. These works assumed a setting with two types of features, $a$ and $b$, with detectors corresponding to the two types situated on a 1D or 2D grid. The binary outputs of these features are represented by an array of $n$ positions:
$a_1,...,a_n$ (where $a_i=1$ if the feature $a$ is detected in position $i$, and $a_i=0$ otherwise) and an array $b_1,...,b_n$ which is analogously defined with respect to $b$. The Boolean function $P_{LR}^{n}$ outputs $1$ if there exist $i,j$ with $i<j$ such that $a_i=b_j=1$, and $0$ otherwise. This function is advocated in \cite{legenstein2002neural} as a simple example of a detection problem in vision that requires to identify spatial relationship among features. It is shown that this problem can be solved by $O(\log n)$ threshold gates. A 2-dimensional analogue where the indices $i=(i_1,i_2)$ and $j=(j_1,j_2)$ represent two-dimensional coordinates and one is interested whether there exist indices $i$ and $j$ such that $a_i=b_j=1$ and $j$ is above and to the right of the location $i$ is studied in \cite{legenstein2002neural}. Recently, the two-dimensional version was studied in \cite{uchizawa2015threshold} where a $O(\sqrt{n})$-gate threshold implementation was given along with a lower bound of $\Omega(\sqrt{n/\log n})$ for the size of any threshold circuit for this problem. We remark that the problem studied in \cite{legenstein2001foundations,legenstein2002neural,uchizawa2015threshold} is different from ours, and different proof ideas are needed for establishing lower bounds in our setting.

%There are several algorithms that solve the string matching problem in $O(n)$ time \cite{boyer1977fast,knuth1977fast}. In \cite{galil1985optimal}, it is mentioned (without a proof) that the parallel algorithm described there implies an $O(n \log^2 n)$ DeMorgan circuit of depth $O(\log ^2 n)$.

%It is also noted~ \cite{galil1985optimal} that the classical Boyer-Moore and Knuth-Morris-Pratt algorithms do not seem to parallelize, thus implementing these algorithms by small depth circuit of size $O(n)$ may be difficult or impossible.

%%%%%%%%%%%%%%%%%%%%%%%%%%%%%%%%%%%%%%%%%%%%%%%%%%%%%%%%%%%%%
\paragraph{Learning patterns.}
%%%%%%%%%%%%%%%%%%%%%%%%%%%%%%%%%%%%%%%%%%%%%%%%%%%%%%%%%%%%%

The language of all strings (of arbitrary length) containing a fixed pattern is regular and can be recognized by a finite automata. There is a large literature on learning finite automata (e.g., \cite{angluin1987learning,freund1997efficient,ron1997exactly}). This literature is mostly concerned with various active learning models and it does not imply our bounds on the sample complexity of learning $\mathcal{H}_{k,\Sigma}$.

Motivated by computer vision applications, several works have considered the notion of \emph{visual concepts}: namely a set of shapes that can be used to classify images in the PAC-learning framework \cite{kushilevitz1996learning,shvaytser1990learnable}. Their main idea is that occurrences of shapes (such as lines, squares etc.) in images can be used to classify images and that furthermore the representational class of DNF's can represent occurrences of shapes in images. For example, it is easy to represent the occurrence of a fixed pattern of length $k$ in a string of size $n$ as a DNF with $n-k$ clauses (see e.g., Lemma~\ref{thm:depth2 construction}). We note that these works do not study the VC dimension of our pattern matching problems (or VC bounds in general). We also observe that no polynomial algorithm is known for learning DNF's and that there is some evidence that the problem of learning DNF is intractable \cite{daniely2016complexity}. Hence the result in \cite{kushilevitz1996learning,shvaytser1990learnable} do not imply that our pattern learning problem (represented as a DNF) can be done in polynomial time.

%\newpage
%%%%%%%%%%%%%%%%%%%%%%%%%%%%%%%%%%%%%%%%%%%%%%&&&&&&&
\section{Communication Complexity}\label{sec:CC}
%%%%%%%%%%%%%%%%%%%%%%%%%%%%%%%%%%%%%%%%%%%%%%&&&&&&&

In this section we prove \Cref{thm:comm}, and also discuss the possibility of a better upper bound.
\comm*

\subsection{Periods in strings} \label{sec:periods}

We say a string $x\in\{0,1\}^n$ has \emph{period $p\in\{0,1\}^i$ of order $i$} if $x$ is a prefix of a high enough power $p^m$ (for some $m\geq 1$). Equivalently, $x$ has a period of order $i$ iff $x[i+1,n]=x[1,n-i-1]$. A classic lemma characterizes the orders of short periods in a string.
\begin{lemma}[{\cite{lyndon62equation}}] \label{lem:period}
If $x$ has periods of orders $i,j$, $i+j\leq|x|$, then there is one of order $\mathrm{gcd}(i,j)$.
\end{lemma}
In particular, all periods of order $\leq n/2$ are powers of some \emph{primitive period} (shortest period of order $\leq n/2$). It is natural to ask: how many bits of communication are required to decide whether a string has a primitive period? We will discuss this in \Cref{sec:evidence}.

%%%%%%%%%%%%%%%%%%%%%%%%%%%%%%%%%%%%%%%%%%%%%%%%%%%%%
\subsection{Upper bound}
%%%%%%%%%%%%%%%%%%%%%%%%%%%%%%%%%%%%%%%%%%%%%%%%%%%%%

We start by describing an $O(\log k\cdot n/k)$-bit deterministic protocol for $\PM$ assuming the pattern $y$ is fixed (known to both players). This immediately gives a protocol of cost $O(k+\log k\cdot n/k)$ when $y$ is \emph{not} fixed: Alice and Bob simply exchange all bits of the $k$-bit pattern and then run the protocol that assumes $y$ is fixed. When $k\leq \sqrt{n}$ this yields the first upper bound claimed in \Cref{thm:comm}.
\begin{lemma} \label{lem:fixed-y}
For every fixed pattern $y\in\{0,1\}^k$ the function $x\mapsto \PM(x,y)$ admits a deterministic protocol of cost $O(\log k\cdot n/k)$ under any bipartition of the input $x$.
\end{lemma}
\begin{proof}
Since every occurrence of pattern $y$ in $x$ must start in one of the $n/k$ many intervals $[1,k],[k,2k],\ldots$, it suffices to to design a $O(\log k)$-bit protocol to test whether $y$ occurs starting in a particular interval, and then repeat this protocol for every interval. Let us describe a protocol for the first interval $[1,k]=[k]$.

Suppose Alice is given the bits $x_I$ for $I\subseteq[n]$ and Bob the bits $x_{\bar{I}}$ for $\bar{I}\coloneqq[n]\smallsetminus I$. The protocol proceeds as follows. First, Alice sends two indices $i,j\in[k]$ where $i$ (resp.\ $j$) is the smallest (largest) index such that it is consistent with Alice's bits $x_I$ that $y$ could appear in $x$ starting at position~$i$~($j$). (If there are no such indices, then the players may output ``no match''.) From $i$ Bob can infer all Alice's bits in the interval $[i,i+|y|-1]$ (the bits agree with $y$, which is known to Bob), and similarly from $j$ Bob can infer Alice's bits in $[j,j+|y|-1]$. Altogether Bob learns Alice's bits in $[i,i+|y|-1]\cup [j,j+|y|-1]=[i,j+k-1]$. Together with his own bits $x_{\bar{I}}$ Bob can then determine whether $y$ occurs in $x$ with a starting position in $[k]$. The cost of the protocol (sending the two indices and the final output value) is $2\log k+1$.
\end{proof}

Next we supply the protocol for the second upper bound in \Cref{thm:comm}.
\begin{lemma} \label{lem:cc-ub}
For $k\geq\sqrt{n}$ the function $\PM$ admits a randomized protocol of cost $O(\log n\cdot \sqrt{n})$ under any bipartition of the input $(x,y)$.
\end{lemma}
\begin{proof}
At the start of the protocol, the two players exchange the first $2\sqrt{n}$ many bits of $y$ so that they both learn the prefix $p\coloneqq y[1,2\sqrt{n}]$. We think of $p$ as fixed from now on. Since any occurrence of $y$ in $x$ must start in one of the $\sqrt{n}$ many intervals $[1,\sqrt{n}],[\sqrt{n},2\sqrt{n}],\ldots$ it suffices to design a $O(\log n)$-bit protocol (with error probability $\leq 1/n$) to test whether $y$ starts in a particular interval, and then repeat this protocol for every interval (resulting in error probability $\leq \sqrt{n}/n$ by a union bound). Let us describe a protocol for the first interval $[1,\sqrt{n}]=[\sqrt{n}]$.

For simplicity of presentation, we first assume that $p$ has no period of order $\leq \sqrt{n}$. We will handle a $p$ with short periods later.

\paragraph{No short period.}
The protocol to test if $y$ occurs in $x$ starting at a position in $[\sqrt{n}]$ is similar to the one in \Cref{lem:fixed-y}. Assuming Alice is given $x_I$ and Bob is given $x_{\bar{I}}$, Alice first sends two indices $i,j\in[\sqrt{n}]$ where~$i$ (resp.\ $j$) is the smallest (largest) index such that it is consistent with Alice's bits that the prefix $p$ could appear in $x$ starting at position $i$~($j$). Bob can again reconstruct all Alice's bits in the interval $[i,j+|p|-1]$ and determine whether $p$ occurs in $x$ with a starting position in $[\sqrt{n}]$. Since we are assuming that $p$ has no period of order $\leq \sqrt{n}$, Bob can find at most one such starting position, say at coordinate $\ell\in[\sqrt{n}]$. (If there is no starting position for the prefix, there is none for the full pattern $y$ and we may output ``no match''.) The remaining goal becomes to test whether $x[\ell,\ell+k-1]=y$. Consider any $i\in[k]$; either
\begin{enumerate}[label=(\arabic*),noitemsep]
\item Alice (or Bob) owns both $x_{\ell-1+i}$ and $y_i$;
\item Alice owns $x_{\ell-1+i}$ and Bob owns $y_i$ (or vice versa).
\end{enumerate}
For coordinates of type (1), the players may test for equality without communication. For coordinates of type (2), the players can execute a randomized test for equality---a single test for all type-(2) coordinates at once---for which there is a well-known $O(\log n)$-bit protocol (with error probability $\leq 1/n$)~\cite[Example 3.5]{kushilevitz97communication}. This concludes the description of the $O(\log n)$-bit protocol (for a $p$ without short periods).

\paragraph{Short period.}
Suppose $p\in\{0,1\}^{2\sqrt{n}}$ has a period of order $\leq \sqrt{n}$. Since the players know $p$, they can both agree on the shortest one (the primitive period), call it $\bar{p}$, $|\bar{p}|\leq\sqrt{n}$.

The players then proceed to find the largest number $m$ such that $\bar{p}^m$ is a prefix of $y$. To do this, Alice (resp.\ Bob) reports the largest $m_A$ ($m_B$) such that it is consistent with her (his) knowledge of the bits of $y$ that $\bar{p}^{m_A}$ ($\bar{p}^{m_B}$) is a prefix of $y$. Then $m\coloneqq \min(m_A,m_B)$ is the sought number. This takes $O(\log n)$ bits of communication.

Next, the players can check, with constant communication, whether $y$ is simply a prefix of $\bar{p}^{m+1}$. If yes, both players would fully know $y$ and hence they can run the protocol from \Cref{lem:fixed-y}. Assume otherwise henceforth. In this case the players can find a string $q$, $|q|\leq |\bar{p}|$, that is not a prefix of $\bar{p}$, and such that $p'\coloneqq \bar{p}^mq$ is a prefix of $y$. This takes $|q|\leq|\bar{p}|\leq\sqrt{n}$ bits of communication.

We claim that $p'$ has no period of order $\leq \sqrt{n}$. This claim would finish the proof, as the players can finally run the \emph{no-short-period} protocol with $p'$ in place of $p$ (note that the cost of that protocol does not depend on $|p|$). To prove the claim, suppose for contradiction that $p'$ (and hence $p$) has a period $\hat{p}$ of order $|\hat{p}| \leq \sqrt{n}$. Since $\bar{p}$ is the primitive period for $p$, $\hat{p}$ must be a power of $\bar{p}$. Therefore $p'$ is a power of $\bar{p}$. But this contradicts our definition of $p'=\bar{p}^mq$.
\end{proof}

\begin{remark}
 For $k\geq \sqrt{n\log n}$ the above protocol can be optimized to have cost $O(\sqrt{n\log n})$. Namely, consider a prefix $p$ (and intervals) of length $\Theta(\sqrt{n\log n})$ rather than $\Theta(\sqrt{n})$.
\end{remark}

%%%%%%%%%%%%%%%%%%%%%%%%%%%%%%%%%%%%%%%%%%%%%%%%%%%%%
\subsection{Lower bound}
%%%%%%%%%%%%%%%%%%%%%%%%%%%%%%%%%%%%%%%%%%%%%%%%%%%%%
Next we prove a lower bound of $\Omega(\log\log k \cdot n/k)$, for every $k\leq n$, on the randomized communication complexity of $\PM$. As a warm-up, we first observe that a reduction from the ubiquitous set-disjointness function yields a randomized lower bound of $\Omega(n/k)$ for $\PM$. We then show how to improve this by a factor of $\log\log k$.

Recall that in the $m$-bit set-disjointness problem, Alice is given $a \in \{0,1\}^m$, Bob is given $b \in \{0,1\}^m$, and their goal is to compute
$\Disj_m(a,b)\coloneqq(\OR_m\circ\AND_2)(a,b)=\bigvee_{i \in [m]}(a_i \wedge b_i)$. It is well known that this function has communication complexity $\Omega(m)$ even against randomized protocols~\cite{kalyanasundaram1992probabilistic,razborov1992distributional,bar2004information}.

\begin{observation} \label{obs:disj}
$\Disj_{\Omega(n/k)}$ reduces to $\PM$ (under some bipartition of input bits).
\end{observation}
\begin{proof}
Given inputs $(a,b)$ of $\Disj_{m}$ to Alice and Bob they construct, without communication, inputs to $\mathrm{SM}_{m(k+1),k}$ as follows. We set $y\coloneqq 1^k$ and
\begin{equation*}
x~\coloneqq~a_1b_{1}1^{k-2}0a_2b_{2}1^{k-2}0 \ldots a_{n}b_{n}1^{k-2}0.
\end{equation*}
This also implicitly determines the bipartition of input bits of $\mathrm{SM}_{m(k+1),k}$; namely, Alice gets all the coordinates of $x$ with $a_i$s, Bob gets those with $b_i$s, and the rest can be split arbitrarily. It is straightforward to check that $\Disj_{m}(a,b)=\mathrm{SM}_{m(k+1),k}(x,y)$.
\end{proof}

To improve the above, we give a reduction from a slightly harder function, $\OR_m \circ \GT_\ell\colon [\ell]^m\times[\ell]^m\to\{0,1\}$, which maps $(a,b)\mapsto \bigvee_{i\in[m]}\GT(a_i,b_i)$ where $\GT_\ell\colon[\ell]\times[\ell]\to\{0,1\}$ is the \emph{greater-than} function given by $\GT_\ell(a,b)\coloneqq 1$ iff $a\geq b$. The claimed lower bound $\Omega(\log\log k \cdot n/k)$ for $\PM$ follows from the following two lemmas. As mentioned in the introduction, \Cref{lem:cc-lb} was conjectured by \cite{watson18communication}.
\begin{lemma} \label{lem:reduction}
$\OR_{\Omega(n/k)} \circ \GT_{\Omega(k)}$ reduces to $\PM$ (under some bipartition of input bits).
\end{lemma}
\begin{lemma} \label{lem:cc-lb}
$\OR_m\circ\GT_\ell$ has randomized communication complexity $\Omega(m\cdot\log\log\ell)$ for any $m$, $\ell$.
\end{lemma}

\begin{proof}[Proof of \Cref{lem:reduction}]
It suffices to describe a reduction from $\GT_k$ to $\mathrm{SM}_{4k,2k+2}$ as this reduction can be repeated $\Omega(n/k)$ times in parallel on disjoint inputs (similarly as in the proof of \Cref{obs:disj}). Given inputs $(a,b)\in[k]\times[k]$ to $\GT_k$ the two players construct inputs $(x,y)$ to $\mathrm{SM}_{4k,2k+2}$ as follows. As before, we set $y\coloneqq 1^{2k+2}$. As for $x$, Alice will own the even coordinates $I\coloneqq \{2,4,\ldots,4k\}$ of $x$ and Bob the odd coordinates $\bar{I}\coloneqq[4k]\smallsetminus I$. Alice sets $x_I\coloneqq 1^{k+a}0^{k-a}$ and Bob sets $x_{\bar{I}}\coloneqq 0^b1^{2k-b}$. The longest all-$1$ pattern in $x$ is then of length $2(k+a-b+1)$, as illustrated below.
\begin{equation*}
\begin{tabular}{r}
$k=5$\\
$a=b=2$
\end{tabular}
\qquad\text{\huge$\leadsto$}\qquad
x\,\coloneqq
\begin{tabular}{ll}
\text{\ttfamily \ 1\ 1\ 1\ 1\ 1\ 1\ 1\ 0\ 0\ 0}& \text{\small (Alice's bits $x_I$)}\\[-1mm]
\text{\ttfamily  0\ 0\smash{$\underbrace{\text{\ttfamily\ 1\ 1\ 1\ 1\ 1\ 1}}_{2(k+a-b+1)}$}\ 1\ 1}  & \text{\small (Bob's bits $x_{\bar{I}}$)}
\end{tabular}\\[4mm]
\end{equation*}
Note that $2(k+a-b+1) \geq 2k+2$ iff $a\geq b$. Hence $\GT_k(a,b)=\mathrm{SM}_{4k,2k+2}(x,y)$, as desired.
\end{proof}

\begin{proof}[Proof of \Cref{lem:cc-lb}]
A standard technique for proving randomized communication lower bounds for functions of the form $\OR_m\circ F$, where $F\colon \mathcal{X}\times \mathcal{Y}\to\{0,1\}$, is \emph{information complexity} (IC)~\cite{chakrabarti01informational,bar2004information,braverman12interactive}. We explain how to combine existing methods to obtain the desired lower bound when $F=\GT_\ell$. Our discussion assumes familiarity with the IC technique.

The usual plan is to exhibit a \emph{one-sided} distribution $\mu_0$ over $F^{-1}(0)$ and prove, for any bounded-error protocol $\Pi$ computing $F$, a lower bound on the \emph{information cost} $I_{\mu_0}\coloneqq \mathbb{I}(\Pi(X,Y):X \mid Y)+\mathbb{I}(\Pi(X,Y):Y \mid X)$ where $XY\sim\mu_0$ (see~\cite{braverman12interactive} for details on information cost). Bar-Yossef et al.~\cite{bar2004information} proved that the randomized communication complexity of $\OR_m\circ F$ is at least $m\cdot I_{\mu_0}$. Hence our goal is to show, for some one-sided $\mu_0$ over $\GT_\ell^{-1}(0)$,
\begin{equation}\label{eq:ic}
I_{\mu_0}~\geq~\Omega(\log\log\ell).
\end{equation}
Braverman and Weinstein~\cite{braverman16discrepancy} already obtained a lower bound like \eqref{eq:ic} except for a \emph{two-sided} distribution $\mu$ over $F^{-1}(0)\cup F^{-1}(1)$. Here we observe that their proof, virtually unchanged, gives the same lower bound also for a one-sided distribution $\mu_0$.

\paragraph{BW simulation.}
Let us summarize the main technical result of \cite{braverman16discrepancy}. They show a general simulation of any bounded-error, say $\leq 1\%$, protocol $\Pi$ computing $F$ with information cost $I_\mu$ relative to a $\mu$ by an ``unbounded-error'' protocol $\Pi'$ (of communication cost $O(I_\mu)$) satisfying the following: With high probability, say $\geq 99\%$, over $(x,y)\sim\mu$, the simulation is \emph{``successful''} (event $\mathcal{Z}$ in the proof of~\cite[Thm~2]{braverman16discrepancy}) meaning that, for some $\delta \coloneqq 2^{-O(I_\mu+1)}$,
\begin{equation} \label{eq:succ} \textstyle
\forall (x,y)\in \mathcal{Z}:\qquad
\mathbb{P}_{\text{coins of $\Pi'$}}[\,\Pi'(x,y)\text{ outputs }F(x,y)\,]~\geq
~\frac{1}{2}+0.9\cdot\delta.
\end{equation}
A crucial property is that even if the simulation fails for an input $(x,y)\notin \mathcal{Z}$, we are still guaranteed that $\Pi'$ does not output the wrong answer with too high a probability \cite[Prop~2]{braverman16discrepancy}:
\begin{equation} \label{eq:fail} \textstyle
\forall(x,y):\qquad\mathbb{P}_{\text{coins of $\Pi'$}}[\,\Pi'(x,y)\text{ outputs }F(x,y)\,]~\geq~
\frac{1}{2}-0.1\cdot\delta.
\end{equation}
By averaging over $(x,y)\sim\mu$ it follows that
\begin{align} \label{eq:bias}
\mathbb{P}_{(x,y)\sim\mu,\,\text{coins of $\Pi'$}}[\,\Pi'(x,y)\text{ outputs }F(x,y)\,]
~&\textstyle\geq~
\frac{1}{2}+(99\%\cdot0.9-1\%\cdot0.1)\delta \notag\\
~&\textstyle\geq~
\frac{1}{2}+0.8\cdot\delta.
\end{align}
In words, $\Pi'$ achieves a non-trivial bias in guessing $F$ relative to $\mu$. The authors conclude~\cite[Thm~1]{braverman16discrepancy} that $\Pi'$ witnesses an $O(\log\delta^{-1})=O(I_\mu+1)$ \emph{discrepancy bound} for $F$ relative to $\mu$. Finally, they provide an $\Omega(\log\log \ell)$ discrepancy bound for $F=\GT_\ell$ relative to a two-sided $\mu$, which proves~\eqref{eq:ic} (except for a two-sided $\mu$).

\paragraph{Our modification.}
Our observation is that the BW simulation can be applied while assuming only an upper bound on $I_{\mu_0}$ for every one-sided $\mu_0$, and still conclude \eqref{eq:bias} for any two-sided $\mu$. Indeed, let $\mu$ be any two-sided distribution; we may assume wlog that it is balanced, $\mu=\frac{1}{2}\mu_0+\frac{1}{2}\mu_1$, where $\mu_b$ is over $F^{-1}(b)$. Suppose $\Pi$ is a protocol for $F$ with information cost $I_{\mu_0}$ relative to $\mu_0$. Then from the BW simulation we can obtain $\Pi'$ such that for some $\delta\coloneqq 2^{-O(I_{\mu_0}+1)}$,
\begin{align*}
\mathbb{P}_{(x,y)\sim\mu_0,\,\text{coins of $\Pi'$}}[\,\Pi'(x,y)\text{ outputs }0\,]~&\textstyle\geq~\frac{1}{2}+0.8\cdot\delta,\\
\mathbb{P}_{(x,y)\sim\mu_1,\,\text{coins of $\Pi'$}}[\,\Pi'(x,y)\text{ outputs }1\,]~&\textstyle\geq~\frac{1}{2}-0.1\cdot\delta,
\end{align*}
where the first bound is from \eqref{eq:bias} (specialized to $\mu_0$) and the second bound is from the failure guarantee~\eqref{eq:fail}. We may finally define a third protocol $\Pi''$ with a slightly scaled-down probability of outputting $0$: $\Pi''(x,y)$ outputs 1 with probability $\delta/2$ and with the remaining probability $1-\delta/2$ it runs $\Pi'(x,y)$. This protocol satisfies \eqref{eq:bias}, albeit with a slightly smaller coefficient than $0.8$.
\end{proof}

\subsection{A better protocol?} \label{sec:evidence}

As bonus results, we give some evidence for the existence of an improved randomized protocol for $\PM$ when $k$ is large. We first define what \emph{unambiguous randomized} (aka $\UBPP$, or unambiguous Merlin--Arthur) protocols are; they generalize the notion of unambiguous deterministic protocols (aka $\UP$) introduced by Yannakakis~\cite{yannakakis91expressing}.

\begin{definition}[$\UBPP$ protocols]
An \emph{unambiguous randomized} protocol $\Pi$ computes a function $F(x,y)$ as follows. In the first phase the players nondeterministically guess a witness string $z\in\{0,1\}^{c_1}$, and then in the second phase they run a randomized (error $\leq 1/3$) protocol of cost $c_2$ to decide whether to accept the witness $z$. The correctness requirement is that for every $(x,y)\in F^{-1}(1)$ there needs to be a unique witness that is accepted; for every $(x,y)\in F^{-1}(0)$ no witness should be accepted. The cost of $\Pi$ is defined as $c_1+c_2$.
\end{definition}

Unambiguous randomized protocols have not been studied before in communication complexity. However, the recent breakthrough of Chattopadhyay et al.~\cite{chattopadhyay19log} (who disproved the log-approximate-rank conjecture of~\cite{lee09lower}) is closely related. It is not hard to see that the function $F(x,y)$ they study (of the form $\mathrm{Sink}\circ\mathrm{XOR}$) admits an $O(\log n)$-cost $\UBPP$ protocol. The authors proved that the usual randomized (aka~$\BPP$) communication complexity of $F$ is high, $n^{\Omega(1)}$. Consequently, there is no generic simulation of a $\UBPP$ protocol by a $\BPP$ protocol. By contrast, Yannakakis~\cite[Lemma~1]{yannakakis91expressing} showed that $\U\mydot\P$ protocols can be made deterministic efficiently.

Our first bonus result is an efficient $\UBPP$ protocol for determining if a given string has a primitive period. We do not know whether there is an efficient randomized protocol.
\begin{lemma} \label{lem:cc-period}
Suppose the bits of $x\in\{0,1\}^n$ are split between two players. There is an $\UBPP$ protocol of cost $O(\log^2 n)$ for deciding whether $x$ has a primitive period (and to compute its order).
\end{lemma}
\begin{proof}
Suppose Alice is given the bits $x_I$, $I\subseteq[n]$, and Bob the bits $x_{\bar{I}}$, $\bar{I}\coloneqq[n]\smallsetminus I$. The idea is that Alice and Bob guess the order of the primitive period, and then verify their guess using randomness. The guess is just a $\log n$-bit number $k\in[n/2]$ having some prime factorization $k=p_1^{e_1}p_2^{e_2}\cdots p_\ell^{e_\ell}$ where $e_i\geq 1$ and $\ell\leq \log n$. In the randomized checking phase Alice and Bob run an $O(\log n)$-bit equality protocol (as in \Cref{lem:cc-ub}) to check whether $x[1,i]$ is a period (namely, they test the equality $x[i+1,n]=x[1,n-i-1]$). If yes, we continue to check that there is no shorter period. Since the shortest (primitive) period divides $k$, it suffices to check that none of the candidates $\{k/p_1, k/p_2,\ldots,k/p_\ell\}$ is a period. For each such candidate we run an equality protocol. Altogether this checking phase costs $O(\ell\log n)=O(\log^2 n)$ bits of communication. The protocol is indeed unambiguous since the primitive period (should it exist) is unique.
\end{proof}

\newcommand{\Rpf}{R_{\mathrm{pf}}}

If we let $\Rpf$ denote the randomized communication complexity of the above period finding problem, then we can interpret \Cref{lem:cc-period} as evidence that $\Rpf\leq \polylog(n)$. Assuming period finding is indeed easy, we can then provide similar evidence for the easiness of $\PM$ for large $k$.
\begin{lemma}\label{lem:PM-UBPP}
$\mathrm{SM}_{n,0.9n}$ admits an $\UBPP$ protocol of cost $O(\log n)+\Rpf$.
\end{lemma}
\begin{proof}
The idea is that the players guess a position $i\in[n]$ ($\log n$ bits) and then verify, using randomness, that $i$ is the starting point for the \emph{earliest} occurrence of $y$ in $x$. More precisely, in the verification phase, the players first run the $\Rpf$-bit protocol to decide whether $y$ has a primitive period (and compute its order). Observe that if $y$ occurs more than once in $x$, then since $k=0.9n$, the occurrences must overlap by $\geq 0.8n$ positions. In this case~$y$ has a period of order $\leq 0.1n\leq k/2$, and hence $y$ has a primitive period. Two cases:
\begin{itemize}[label=$-$]
\item \emph{$y$ does not have a primitive period.}
Then $y$ can appear at most once in $x$. The players run an $O(\log n)$-bit equality protocol (as in \Cref{lem:cc-ub}) to test whether $y$ starts at position $i$ in $x$.
\item \emph{$y$ has a primitive period of order $\ell\in[k/2]$.}
Then position $i$ is the earliest occurrence of $y$ in $x$ iff (1) $y$ starts at position $i$ in $x$, and (2) $y$ does not start at position $i-\ell$ in $x$. The conditions (1) and (2) can be checked by running an equality protocol twice.
\end{itemize}
\end{proof}

\section{Threshold Circuits}
In this section we prove \Cref{thm:threshold}.
\threshold*
In \Cref{sec:threshold construction} we prove the upper bound, in \Cref{sec:lb threshold unbdd k = logn} we give the lower bounds. Finally, in \Cref{sec:thrdepth2} we study the complexity of $\PM$ in the models of restricted threshold circuits.

%%%%%%%%%%%%%%%%%%%%%%%%%%%%%%%%%%%%%%%%%%%%%%%%%%%%%%%%%%%%%
\subsection{Upper bound}\label{sec:threshold construction}
%%%%%%%%%%%%%%%%%%%%%%%%%%%%%%%%%%%%%%%%%%%%%%%%%%%%%%%%%%%%%
We start with a construction giving the upper bound of \Cref{thm:threshold}.

\begin{lemma}\label{thm:upthdepth2}
There is a depth-$2$ threshold circuit of size $O(n-k)$ computing $\PM$.
\end{lemma}
\begin{proof}
Let $\GEQ$ be the gate that gets $2k$ bits $z_1,\dots,z_k ; w_1,\dots,w_k$
and evaluates to 1
if and only if the number represented in binary by the bits $z_1,...,z_k$
is greater than or equal to the number represented by $w_i,...w_k$.
Note that $\GEQ$ can be implemented by one threshold gate
as follows: $\GEQ(z_1,\dots,z_k ; w_1,\dots,w_k) =1$ if and only if
$(z_1 + 2\cdot z_2+ 4\cdot x_3+...+2^{k-1} \cdot z_k ) - (w_1 + 2\cdot w_2+ 4\cdot w_3+...+2^{k-1} \cdot w_k ) \geq 0$.

Now we describe a circuit computing $\PM(x_1,\dots,x_n ; y_1,\dots,y_k)$.
Let the first layer contain $n-k+1$ $\GEQ$ gates $g_i$ and $n-k+1$ gates $\ell_i$ for $i=1,\dots n-k+1$,
where each $g_i$ gets as inputs $x_i,\dots,x_{i+k-1} ; y_1,\dots,y_k$ (and evaluates to $1$
if and only if the number represented by the corresponding bits of $x$ is \emph{at least} that represented by $y$),
and $\ell_i$ gets as inputs $y_1,\dots,y_k ; x_i,\dots,x_{i+k-1}$ (and evaluates to $1$
if and only if the number represented by the corresponding bits of $x$ is \emph{at most} that represented by $y$).
The second (output) layer evaluates to $1$ if and only if $\sum_{i=1}^{n-k+1} (g_i+\ell_i) \geq n-k+2$.
Clearly the circuit contains $2n-2k+3$ gates.

In order to prove correctness, we note that for every $i$, at least one of $g_i$ and $\ell_i$ evaluates to $1$,
and $x[i,i+k-1] = y$ if and only if both $g_i$ and $\ell_i$ are equal to $1$.
Therefore, $\sum_{i=1}^{n-k+1} (g_i+\ell_i) > n-k+1$ if and only if $y$ is a substring of $x$, i.e., $\PM(x,y) = 1$.
\end{proof}

\subsection{Lower bounds}\label{sec:lb threshold unbdd k = logn}
In order to prove the first lower bound of $\Omega(\frac{n\log\log{k}}{k\log{n}})$ we use the classical result on communication complexity of threshold gates~\cite{nisan1993communication}, and the lower bound on communication complexity of $\PM$ from \Cref{thm:comm}.

Nisan and Safra~\cite{nisan1993communication} proved that for \emph{any} bipartition of the $n$ input bits, the $\epsilon$-error randomized communication complexity of a threshold gate (with arbitrary weights) has communication complexity $O(\log{n/\epsilon})$. From this they concluded that for any function $f$, a lower bound of $m$ on the randomized communication complexity for \emph{some} bipartition of the input implies a lower bound of $\Omega(m/\log{n})$ on the threhold complexity of $f$. Now the lower bound of $\Omega(n\log\log{k}/k)$ from \Cref{thm:comm} implies the lower bound of $\Omega(\frac{n\log\log{k}}{k\log{n}})$ on the size of an unbounded depth threshold circuit computing $\PM$.

Below we prove the second lower bound stated in \Cref{thm:threshold}.
The lower bound is shown via a reduction from a hard function $f \colon \cc[k/2-1] \to \zo$ which has $n/k$ preimages of $1$: $|f^{-1}(1)|=n/k$. First, we prove the desired lower bound for the case where $k$ is even and $n$ is a multiple of $k$. In the end of this section we explain how to adjust the proof to the remaining cases. Let $\ell$ and $t$ be integers such that $k=2\ell+2$ and $n=t\cdot k$. Let $F_{\ell,t} = \{f \colon \cc[\ell] \to \zo : |f^{-1}(1)|=t\}$ be the class of Boolean functions of $\ell$ inputs which have exactly $t$ preimages of $1$.

We prove this lower bound via a reduction from a hard function $f \in F_{\ell,t}$.
Specifically, we show that if $\PM$ can be solved by a circuit of size $s$,
then every function $f \in F_{\ell,t}$ also has a circuit of size $s$ computing it.
Then, we show that there are functions in $F_{\ell,t}$
that require large threshold circuits, which implies the corresponding lower bound for the $\PM$ function.

\paragraph{The reduction.}
Given a string $a \in \cc[\ell]$ define $dup(a) \in \cc[k]$
to be the string obtained from $a$
by repeating each bit of $a$ twice, and concatenating it with $01$ in the end.
(Note that $2\ell + 2 = k$ by the choice of $\ell$).
For example $dup(010) = 00110001$.
%Also, for an integer $\ell$ define $z(\ell) \in \cc[k]$ to be the string $0101\dots01$ of length $2\ell+2$,
%i.e., the string where the substring $01$ is repeated $\ell+1$ times.

%Let $f \colon \cc[\ell] \to \zo$ be an arbitrary function.
%For each $a =(a_1,\dots,a_\ell)\in \cc[\ell]$ an input to $f$, define
%\begin{equation*}
%    x_a = \begin{cases}
%            dup(a), & \mbox{if } f(a) = 1; \\
%            z(\ell), & \mbox{if } f(a) = 0.
%          \end{cases}
%\end{equation*}

\begin{observation}\label{observ:reduction to PM t-sparse}
Given a function $f \in F_{\ell,t}$
define $x_f \in \cc[tk]$ to be the concatenation of $dup(a)$ for all $a \in f^{-1}(1)$ in the lexicographic order on $\cc[\ell]$.
Note that $|x_f| = tk = n$.
Then, for any $y \in \cc[\ell]$
it holds that $f(y)=1$ if and only if $\PM(x_f,dup(y)) = 1$.
\end{observation}
Indeed, it is immediate to see that if $f(y)=1$ then $\PM(x_f,dup(y)) = 1$. Duplicating every bit in $a$ and adding $01$ to the end
of the resulting pattern are done to ensure that if $f(y)=0$ there will not be a copy of $dup(y)$ in $x_f$.

Given the observation above, it is not difficult to see that any lower bound on the
size of a circuit computing $f\in F_{\ell,t}$ implies a lower bound on $\PM$.
\begin{proposition}\label{prop:reduction threshold t-sparse}
Let $C$ be a threshold circuit computing $\PM$. Then for every $f\in F_{\ell,t}$, there exists a threshold circuit $C'$ computing $f$ such that $|C'| \leq |C|$.
\end{proposition}

\begin{proof}
Suppose there exists a circuit $C$ of size at most $s$ computing $\PM$. We denote the input variables of the pattern $y$
by $y_1,y_2 \ldots y_{k = 2 \ell+2}$.
We show how to convert it into a circuit $C'$ computing $f$ by fixing some of the input variables of $C$.
This is done by
(1) fixing the ``text part" (the variables corresponding to $x$) of the input of $\PM$ to be $x_f$ as defined in \Cref{observ:reduction to PM t-sparse},
and (2) replacing every pair of variables $y_{2i-1}$ and  $y_{2i}$ for all $i=1,\dots,\ell$ by a single variable $\widehat{y}_i$ that is
fed to all gates that have inputs $y_{2i-1}$ or $y_{2i}$ (with a proper adjustment of the weight if both $y_{2i-1}$ and $y_{2i}$ are inputs of the gate).
Finally, fix $y_{2\ell+1} = 0$ and $y_{2\ell+2} = 1$. It is now easy to see that $C'$ computes~$f$.
\end{proof}

In order to complete the proof of \Cref{thm:threshold},
we need to show that there exists a function $f \in F_{\ell,t}$
that requires large threshold circuits. For this, we compare the number of small threshold circuits (see, for example, \cite{jukna2012boolean,kane2016super}) with the number of functions in $F_{\ell,t}$.

\begin{proposition}\label{prop:hard functions for thresholds t-sparse}
    Let $\ell \in \N$ be sufficiently large, and let $t \in \N$.
      There exists a function $f \in F_{\ell,t}$
      such that any threshold circuit (with no restrictions on its depth)
      computing $f$ must be of size at least $\Omega(\sqrt{t - t\log{t}/\ell})$.
\end{proposition}

\begin{proof}

We first upper bound the number of functions that can be represented by threshold circuits of size at most $s$.
This can be obtained using the following result from \cite{roychowdhury1994neural}.

\begin{theorem}\label{thm:ROS}
Let $f_1 \ldots f_s \colon \cc[\ell] \to \zo$ be a set of $s$ Boolean functions.
Then, the number of Boolean functions which are realized by a threshold gate $g \colon \cc[s] \to \zo$
whose $s$ inputs are $f_1 \ldots f_s$ is at most $2^{O(\ell s)}$.
\end{theorem}

It follows from \Cref{thm:ROS} that the number of distinct
Boolean functions with $\ell$ variables computed by a threshold circuit of size $s$ is $2^{O(\ell s^2)}$ (as there are at most $2^{O(\ell s)}$ choices for every gate and there are $s$ gates).
On the other hand, the number of Boolean functions $f \in F_{\ell,t}$
is ${2^\ell \choose t} \geq (\frac{2^\ell}{t})^t = 2^{(\ell - \log{t})t}$.
Therefore, there exists a function $f \in F_{\ell,t}$ that cannot be computed by a threshold circuit
of size $s \geq \Omega(\sqrt{t - t\log{t}/\ell})$.
\end{proof}

We now derive the desired lower bound on the size of threshold circuits computing the string matching function.
%For general (unbounded depth) threshold circuit it follows that any threshold circuit computing $\PM$
%has size at least $\Omega(\sqrt{(\ell - \log(t))t/\ell}) = \Omega(\sqrt{t - t \log(t)/\ell})$.
Plugging in $k = 2\ell + 2$ and $n = tk$,
we get the lower bound of $s \geq \Omega\big(\sqrt{\frac{n}{k}  - \frac{2n}{k^2} \cdot \log(\frac{n}{k})}\;\big)=\Omega(\sqrt\frac{n}{k})$ assuming $k\geq \Omega(\log n)$.

Now we describe how this proof can be adopted for the case when $n$ is not a multiple of $k$ and the case of odd $k$. First, in order to handle the case of pattern of \emph{odd} length, one can add the string $010$ (instead of $01$) to the end of $dup(a)$. If $n$ is not a multiple of $k$, then in the reduction above we can pad the string $x_f$ with zeros in the end,
and the reduction still satisfies the property that
$f(y)=1$ if and only if $\PM(x_f,dup(y)) = 1$
as in \Cref{observ:reduction to PM t-sparse},
and the same lower bound holds (up to a constant factor in the asymptotics).

\subsection{Depth-2 Circuits}\label{sec:thrdepth2}

In \cref{thm:weird_bounds} we prove lower bounds for some restricted classes of depth-$2$ circuits computing $\PM$.
These results should be contrasted with the upper bounds of \Cref{thm:threshold} and \Cref{thm:lb DeMorgan gate elimination}. Namely, there exists an $\LTF\circ\LTF$ circuit of size $O(n-k)$ and an $\OR\circ\AND\circ\OR$ circuit of size $O(nk)$ computing $\PM$.

We recall a few definitions. Let $\ELTF$ denote the class of \emph{exact} threshold functions (that is, the functions which output $1$ on an $m$-bit input $x$ if and only if $\sum_{i\in[m]} a_ix_i= \theta$ for some fixed coefficient vector $a \in \R^{m}$, and $\theta \in \R$). Similarly, $\EMAJ$ denotes the class of exact majorities which output $1$ if and only if the sum of their $m$ Boolean inputs is exactly $m/2$. By $\SYM$ we denote the class of all symmetric Boolean functions. For two classes of functions $\C_1$ and $\C_2$, by $\C_1\circ\C_2$ we denote the class of depth-$2$ circuits where the output gate is from $\C_1$ and the gates of the first layer are from $\C_2$. For a class of circuits $\C$ and a function $f$, be $\C(f)$ we denote the minimal size of a circuit from $\C$ computing $f$.

In proving lower bounds for $\PM$ a simple yet useful property is that Observation~\ref{obs:disj} can
be applied to circuits as well. This allows to reduce the disjointness problem to string matching, and get lower bounds for $\PM$ via known circuit lower bounds for disjointness.
The point is that a circuit $C$ with strings of length roughly $m k$ for $\PM$ (and patterns of length $k$) can be used to solve disjointness on strings of length $m$
by feeding $C$ with the string $x~\coloneqq~a_1b_{1}1^{k-2}0a_2b_{2}1^{k-2}0 \ldots a_{n}b_{n}1^{k-2}0$ and the pattern $y=1^k$. Hence a lower bound of $s(n)$ for circuits
computing disjointness implies a lower bound of $\Omega(s(n/k))$ for circuits computing $\PM$.

\begin{theorem}\label{thm:weird_bounds} For every $1<k\leq n$,
\begin{enumerate}[itemsep=1pt]
\item $\OR\circ\LTF(\PM)\geq \Omega(n-k)$;
\item $\AND\circ\LTF(\PM)\geq 2^{\Omega(n/k)}$;
\item $\AND\circ\OR\circ\XOR(\PM)\geq 2^{\Omega(n/k)}$;
\item $\ELTF\circ\SYM(\PM)\geq 2^{\Omega(n/k)}$;
\item $\EMAJ\circ\ELTF(\PM)\geq 2^{\Omega(n/k)}$.
\end{enumerate}
\end{theorem}
\begin{proof}
%Recall that by \Cref{obs:disj}, $\Disj_{\Omega(n/k)}$ reduces to $\PM$.
\begin{enumerate}
\item We will prove that even for the fixed pattern $y=1^k$, the number of $\LTF$ gates in any $\OR\circ\LTF$ circuit computing $\PM$ must be at  least $(n-k+1)/2$. Assume, for the sake of contradiction that there exist $t<(n-k+1)/2$ threshold gates $g_1,\ldots, g_t$ whose $\OR$ computes $\PM(x,1^k)$. For $0\leq i\leq n-k$, let $x_i=0^i 1^k 0^{n-k-i}$ be a string of length $n$. Note that for every $i$, $\PM(x_i,1^k)=1$, therefore, there exists at least on gate $g_j$ for $1\leq j\leq t$ accepting it. Since there are $n-k+1$ strings $x_i$, and $t<(n-k+1)/2$ gates, at least one gate $g_j$ must accept two non-consecutive $x_i$'s. Without loss of generality assume that $g_1$ accepts $(x_i,y)$ and $(x_j,y)$ for $j>i+1$. Now let
\begin{align*}
x&=0^i10^{j-i-1}1^{k-1}0^{n-k-j+1} \;, \\
x'&=0^{i+1}1^{k-1}0^{j-i-1}10^{n-k-j} \; .
\end{align*}
For the fixed pattern $y=1^k$, suppose $g_1$ computes the function $g_1(x)=\sum_{m=1}^n \alpha_m x[m]\geq \theta$ of the text $x$. From $g_1(x_i)=g_1(x_j)=1$, we have that $\sum_{m=i+1}^{i+k}\alpha_m+\sum_{m=j+1}^{j+k}\alpha_m\geq 2\theta$. Now we apply the function $g_1$ to $x$ and $x'$:
\begin{align*}
g_1(x)+g_1(x')=\left(\alpha_{i+1}+\sum_{m=j+1}^{j+k-1}\alpha_m\right) + \left(\sum_{m=i+2}^{i+k}\alpha_m + \alpha_{j+k}\right)
=\sum_{m=i+1}^{i+k}\alpha_m+\sum_{m=j+1}^{j+k}\alpha_m\geq 2\theta \;.
\end{align*}
Therefore, at least one of the inputs $x$ and $x'$ is accepted by $g_1$ (and, therefore, by the $\OR\circ\LTF$ circuit). Note that $x$ and $x'$ each has $k$ ones with a zero in between (since $j>i+1$). Therefore, neither $x$ not $x'$ can be accepted by a circuit computing $\PM(x,1^k)$, which leads to a contradiction.

\item
We will simulate an $\OR\circ\LTF = \neg\AND\circ\LTF$ circuit by a special type of private-coin \emph{``small bounded-error''} protocol against which lower bounds are known for computing $\neg\Disj_m$ for $m\coloneqq n/k$~\cite{goos16communication} (and hence $\neg\PM$ by \Cref{obs:disj}). Suppose the top fan-in of an $\OR\circ\LTF$ circuit is $t$ and its input $x\in\{0,1\}^m$ is bipartitioned between two players. In the simulation, Alice first uses her private coins to choose a uniform random $i\in[t]$ and sends it to Bob ($\log t$ bits). Then Alice and Bob together evaluate the $i$-th LTF gate to within error $\epsilon\coloneqq 0.1/t$, which takes $O(\log(m/\epsilon))=O(\log m+\log t)$ many bits of communication. This protocol is such that it accepts every $1$-input of the circuit with probability at least $\alpha\coloneqq 1/t \cdot (1-\eps)$ (at least one LTF evaluates to 1); and every $0$-input it accepts with probability at most $\epsilon\leq \alpha/2$. It is known that every such protocol (with a constant-factor acceptance gap, $\alpha$ vs.\ $\alpha/2$) for $\neg\Disj_m$ (and hence for $\neg\PM$) needs $\Omega(m)$ bits of communication~\cite[Thm~1.3]{goos16communication}. This shows a lower bound of $t\geq 2^{\Omega(m)}=2^{\Omega(n/k)}$ for the size of any $\OR\circ\LTF$ circuit for $\neg\PM$, or equivalently, any $\AND\circ\LTF$ circuit for $\PM$.

\item The above proof works with the layer of $\LTF$ gates replaced by $\OR\circ\XOR$ gates: both types of gates admit an $O(\log(n/\epsilon))$-bit $\epsilon$-error protocol (evaluating $\OR\circ\XOR$ involves computing the \emph{equality} function).
\item This bound follows from the reduction from $\Disj$ and Theorem~15 in Hansen and Podolskii~\cite{hansen2010exact}.
\item Follows from the work of Razborov and Sherstov~\cite{razborov2010sign}, and the closedness of $\EMAJ\circ\ELTF$ under $\AND$ (see Theorem~20 in~\cite{hansen2010exact}).
\end{enumerate}
\end{proof}
%%%%%%%%%%%%%%%%%%%%%%%%%%%%%%%%%%%%%%%%%%%%%%
\section{DeMorgan Circuits}\label{sec:depth2}
%%%%%%%%%%%%%%%%%%%%%%%%%%%%%%%%%%%%%%%%%%%%%%
In this section we prove \Cref{thm:depth2lb} and \Cref{thm:lb DeMorgan gate elimination}.
\depthtwolb*

\elimination*
In \Cref{sec:construction DeMorgan unbdd} we give upper bounds for both theorems, in~\Cref{sec:lb DeMorgan d2} we prove lower bounds for depth-$2$ circuits, and in~\Cref{sec:unbounded} we provide a lower bound for the unbounded depth case.

%%%%%%%%%%%%%%%%%%%%%%%%%%%%%%%%%%%%%%%%%%%%%%%%%%%%%%%%%%%%%%%%%%%%%%%%%%%%%%%%%%
\subsection{Upper Bounds}\label{sec:construction DeMorgan unbdd}
%%%%%%%%%%%%%%%%%%%%%%%%%%%%%%%%%%%%%%%%%%%%%%%%%%%%%%%%%%%%%%%%%%%%%%%%%%%%%%%%%%
We first give a $\DNF$ with $2^k(n-k+1)$ clauses computing $\PM$, and in~\Cref{lem:DNF} we will prove that this $\DNF$ is essentially optimal.
%In this section we prove the upper bounds (i.e., the first two items) of \Cref{thm:depth2lb}.

\begin{lemma}\label{thm:depth2 construction}
For any $k \leq n$ there exists a DeMorgan circuit of depth 2 and size $(n-k+1) \cdot 2^k +1$ computing $\PM$.
\end{lemma}

\begin{proof}
First we note that equality of two $k$-bit strings can be implemented using a DNF of width $2k$ and size (number of clauses) $2^k$.
    Indeed, denoting the two inputs by $z = (z_1,\dots,z_k)$ and $w = w_1,\dots,w_k$, let
    \begin{equation*}
        \EQ(z_1,\dots,z_k ; w_1,\dots,w_k) =
        \bigvee_{a = (a_1,\dots,a_k) \in \{0,1\}^k} \big(\wedge_{i=1}^k (z_i=a_i) \wedge_{i=1}^k (w_i=a_i) \big) \enspace,
    \end{equation*}
    where $(w_i=a_i)$ is equal to $w_i$ if $a_i = 1$, and $\neg w_i$ otherwise.

    For each $i=1,\dots,n-k+1$ let $\EQ_i$ be the DNF that outputs 1 if and only if
    $y = (x_i,\dots,x_{i+k-1})$. Taking $\bigvee_{i=1}^{n-k+1} \EQ_i$ we obtain a circuit
    of depth-$3$ that computes the $\PM$ function.
    In order to turn it into a depth-$2$ circuit, note that the second and the third layers
    consist of $\vee$ gates, and hence can be collapsed to one layer.
    This way we get a depth-$2$ circuit of size $(n-k+1) \cdot 2^k +1$.
\end{proof}

Now we show that already in depth $3$, one can compute $\PM$ by a much smaller circuit. This Lemma is likely to have been discovered multiple times, we attribute it to folklore.
\begin{lemma}\label{thm:depth3 construction}
There exists a DeMorgan circuit of depth 3 and size $O(nk)$ computing $\PM$.
\end{lemma}

\begin{proof}
    First we note that the equality function of two $k$-bit strings can be implemented using a $\CNF$ of width $2$ and size (number of clauses) $2k$.
    Indeed, we can check equality of two bits $z$ and $w$ using the circuit $(z \vee \neg w) \wedge (\neg z \vee w)$.
    Therefore, we can implement equality of two $k$-bits strings using the CNF formula

    \[
    \EQ(z_1,\dots,z_k ; w_1,\dots,w_k) = \bigwedge_{i=1}^{k} \left( (z_i \vee \neg w_i) \wedge (\neg z_i \vee w_i) \right) \; .
    \]

    From here on we can proceed as in the previous proof, namely, for each $i=1,\dots,n-k+1$ let $\EQ_i$ be the CNF that outputs 1 is and only if
    $y = (x_i,\dots,x_{i+k-1})$. Taking $\bigvee_{i=1}^{n-k+1} \EQ_i$ we obtain a circuit
    of depth 3 that computes the $\PM$ function.
    The output gate has fanin $n-k+1$, the gates in the second layer have fan-in $2k$,
    and the gates in the first layer have fan-in $2$.
    Therefore, the total size of the circuit is $O(nk)$.
\end{proof}

%%%%%%%%%%%%%%%%%%%%%%%%%%%%%%%%%%%%%%%%%%%%%%%%%%%%%%%%%%%%%%%%%%%%%%%%%%%%%%%%%%
\subsection{Lower bounds for depth 2}\label{sec:lb DeMorgan d2}
%%%%%%%%%%%%%%%%%%%%%%%%%%%%%%%%%%%%%%%%%%%%%%%%%%%%%%%%%%%%%%%%%%%%%%%%%%%%%%%%%%
We may assume wlog that every optimal circuit of depth $2$ is either a $\CNF$ or a $\DNF$. First, we show that in the class of $\DNF$s, the construction from \Cref{thm:depth2 construction} is optimal (up to a constant factor).

\begin{lemma}\label{lem:DNF}
\item For every $k>1$, the $\DNF$-size of $\PM$ is at least
\[
\DNF(\PM) \geq 2^{k-1}(n-k+1) \; .
\]
\end{lemma}
\begin{proof}
Let us consider the set $P$ of $2^{k-1}$ patterns of length $k$ which all start with a $1$:
\[
P=\{1p\colon p\in\{0,1\}^{k-1}\} \; .
\]
For any pattern $p\in P$ and integer $0\leq i\leq n-k$, let $s_{p, i}=0^i p 0^{n-k-i}$ be the string containing $p$ at the $i+1$th position and having zeros everywhere else. Now let the set $S$ be the set of inputs to the $\PM$ problem (that is a set of pairs of a text and pattern) consisting of all $p\in P$ and the corresponding $s_{p,i}$'s:
\[
S=\{(s_{p,i}, p)\colon p\in P, 0\leq i\leq n-k \}
\; .
\]
Consider a $\DNF$ $\phi$ computing $\PM$. In order to show that it has at least $2^{k-1}(n-k+1)$ clauses, we will show that no pair of distinct inputs from $S$ can be accepted by the same clause. Indeed, since every input from $S$ must be accepted by $\phi$ and $|S|=2^{k-1}(n-k+1)$, we get the corresponding lower bound on the number of clauses in $S$.

Assume, for the sake of contradiction that $(s_{p_1,i_1}, p_1)\neq (s_{p_2,i_2}, p_2)$ are accepted by the same clause $C$. Consider the following two cases.
\begin{description}
\item{Case 1: $i_1=i_2, p_1\neq p_2$.} Since $p_1\neq p_2$, there exists an index $2\leq j\leq k$ such that $p_1[j]\neq p_2[j]$. The clause $C$ cannot depend on the $(i_1+j)$th character of the text, because if it depended on it it wouldn't accept one of these strings. Let us consider the string $s$ which differs from $s_{p_1,i_1}$ only in the character number $i_1+j$. Then the input $(s, p_1)$ must still be accepted by $\phi$. This contradicts the definition of $\PM$ because the string $s$ contains exactly one string of length $k$ staring with a 1, and that string differs from $p_1$ in one character.
\item{Case 2: $i_1\neq i_2$.} Wlog assume that $i_1<i_2$. Then the strings $s_{p_1,i_1}$ and $s_{p_2,i_2}$ differ in the character number $j=i_1+1$. (Indeed, by the definition of $p_1$, $s_{p_1,i_1}$ has a 1 in the $j$th position, while $s_{p_2,i_2}$ has a 0 since $i_2>i_1$.) Again, this implies that $C$ does not depend on the $j$th character of the text.
Let us now consider the string $s$ which differs from $s_{p_1,i_1}$ only at the character number $j$. Then the input $(s, p_1)$ is accepted by $\phi$ which leads to a contradiction.
\end{description}
\end{proof}

Now we will prove lower bounds for $\CNF$s computing $\PM$. We will need the following definition.
\begin{definition}\label{def:minterm maxterm}
 A \emph{maxterm} of a Boolean function $f$ is a set of variables of $f$, such that some assignment to those variables makes $f$ output $0$ irrespective of the assignment to the other variables.
The \emph{width} of a maxterm is the number of variables in it.
\end{definition}

First we find the minimal width of maxterms of $\PM$.
\begin{lemma}\label{lem:minterm}
For any $k\leq n$, every maxterm of $\PM$ has width at least
\begin{center}
\begin{tabular}{l l }
\upshape $2\sqrt{n-k+1}$ & for all $k$ ;\\
\upshape $k+\frac{n-k+1}{k}$ & if $k\leq \sqrt{n-k+1}$.\\
\end{tabular}
\end{center}
\end{lemma}

\begin{proof}

Consider a substitution $\rho$ which fixes $n_1$ variables in the text and $k_1$ variables in the pattern.
In order to force $\PM$ to output $0$, for every shift $1\leq i \leq n-k+1$ there must be an index $1\leq j\leq k$ such that $\rho$ assigns a value to $y_j$ and $x_{i+j}$.
Thus, every of $n_1$ assigned variables in the text ``covers'' at most $k_1$ shifts. Since the total number of shifts is $n-k+1$, we have that $n_1\cdot k_1 \geq n-k+1$.
Therefore, by the the inequality of arithmetic and geometric means, $n_1+k_1\geq 2\sqrt{n_1 \cdot k_1} \geq 2\sqrt{n-k+1}$.

Since $n_1\cdot k_1 \geq n-k+1$, we have that $n_1 + k_1 \geq \frac{n-k+1}{k_1} + k_1$.
The second bound follows by noting that the function $f(k_1) = \frac{n-k+1}{k_1} + k_1$ is monotone decreasing for $k_1 < \sqrt{n-k+1}$.

%\end{enumerate}
\end{proof}

Next we prove tight bounds on the number of non-satisfying inputs of $\PM$.
\begin{lemma}\label{lem:preimages}
For $k \leq n$, let $Z$ denote the set of preimages of $0$ of $\PM$.
That is,
%\begin{equation*}
%    O = \{(x,y)\in\{0,1\}^{n+k}:\PM(x;y)=1 \}    \qquad Z = \{(x,y)\in\{0,1\}^{n+k}:\PM(x;y)=0 \}.
%\end{equation*}
\begin{equation*}
Z = \{(x,y)\in\{0,1\}^{n+k}:\PM(x;y)=0 \}.
\end{equation*}

Then
\begin{center}
\begin{tabular}{l l }
\upshape $|Z| = \Theta\left(2^{n+k}\right)$                    & if $k\geq\log{n}+1$;\\
\upshape $|Z| \geq \Omega\left(2^n(1-2^{-k})^n\right)$         & for all $k$.
\end{tabular}
\end{center}

%\begin{tabular}{ r l l }
%& $|O| = \Theta\left(2^n\cdot\min(2^k,n-k+1)\right)$    & for every $k$; \\
%& $|Z| = \Theta\left(2^{n+k}\right)$                    & for $k\geq\log{n}+1$; \\
%& $|Z| \geq \Omega\left(2^n(1-2^{-k})^n\right)$         & for every $k$. \\
%\end{tabular}
\end{lemma}

\begin{proof}
%In order to estimate the number of satisfying inputs of $\PM$, we use the following result from~\cite{gheorghiciuc2007correlation}:
%\begin{theorem}[Corollary 2.2 in~\cite{gheorghiciuc2007correlation}]\label{thm:expectation}
%Let $x$ be a uniformly distributed random string of length $n$, and let $X_{n,k}$ denote the number of \emph{distinct} substrings of $x$ of length $k$. Then there exist $0<\mu<1$ and $0<\eps<0.5$, such that for all values of $n$ and $k$:
%\[
%\E[X_{n,k}] = 2^k-2^k(1-2^{-k})^{n-k+1}+O(n^{-\eps}\mu^k) \; .
%\]
%\end{theorem}

%The upper bound on $|O|$ follows immediately from the following observation. A string of length $n$ can contain at most $\min(2^k,n-k+1)$ different substrings of length $k$. For the lower bound on $|O|$ we consider two regimes of $k$.

%First, if $(n-k+1)/2^k \geq 1$, then from Theorem~\ref{thm:expectation} we have
%\begin{align*}
%|O| &\geq 2^n \cdot \E[X_{n,k}] \geq 2^{n+k} (1-(1-2^{-k})^{n-k+1})\\
% &\geq 2^{n+k} (1 - \exp(-(n-k+1)/2^k)) \geq 2^{n+k}(1-1/e) \;.
%\end{align*}
%If $(n-k+1)/2^k \leq 1$, then again from Theorem~\ref{thm:expectation} we have
%\begin{align*}
%|O| &\geq 2^n \cdot \E[X_{n,k}] \geq 2^{n} (2^k-2^k(1-2^{-k})^{n-k+1})\\
% &\geq 2^{n} \left(2^k - 2^k(1-\frac{n-k+1}{2^k}+\frac{(n-k+1)^2}{2\cdot2^{2k}})\right) \\
% &\geq \Omega(2^{n}(n-k+1)) \;.
%\end{align*}

Let $O$ denote the set of preimages of $1$ of $\PM$. Since every string of length $n$ contains at most $n-k+1$ different substrings of length $k$, we have that $|O|\leq n\cdot 2^n$.
Now for $k\geq\log{n}+1$, from the equation $|Z|+|O|=2^{n+k}$, we have that $|Z|\geq\Omega(2^{n+k})$.

In order to prove the lower bound $|Z|=\Omega\left(2^n(1-2^{-k})^n\right)$, we consider the pattern string $y_0=0^k$. The number $F_n$ of strings $x$ of length $n$ which do not contain $y_0$ satisfies the generalized Fibonacci recurrence:
\[
F_n=\sum_{i=1}^k F_{n-i}\; .
\]
From the known bounds on the generalized Fibonacci numbers (see, e.g., Lemma 3.6 in~\cite{wolfram1996solving}), we have $F_n\geq \Omega(2^n(1-2^{-k})^n)$, which implies the corresponding lower bound on $|Z|$.

\end{proof}

\begin{lemma}\label{thm:CNF DNF}
%We have the following bounds for depth-$2$ DeMorgan circuits for $\PM$.
For every $k$, the $\CNF$-size of $\PM$ is at least
%\begin{enumerate}
%\item The $\DNF$-size of $\PM$ is
%\[
%\DNF(\PM) \geq \Omega\left(2^k\cdot\min(2^k,n-k+1)\right).
%\]
%\item For $\CNF$'s computing $\PM$ we have the following lower bounds on their sizes.

\begin{center}
\begin{tabular}{l l }
\upshape $\CNF(\PM) \geq \Omega\left(2^{\frac{n}{10k}}\right)$     & if  $1<k\leq\log{n}+1$;\\
\upshape $\CNF(\PM) \geq \Omega\left(2^{k+n/k}\right)$                    & if  $\log{n}+1\leq k\leq\sqrt{n}$;\\
\upshape $\CNF(\PM) \geq \Omega\left(2^{2\sqrt{n-k+1}}\right)$            & if  $k\geq\sqrt{n}$.
\end{tabular}
\end{center}
%\end{enumerate}
\end{lemma}
%\Cref{thm:CNF DNF} implies \Cref{thm:depth2lb} since any DeMorgan circuit of depth 2 is either a CNF or a DNF.
\begin{proof}

%First we prove the lower bound for $\DNF$s. From \Cref{lem:minterm}, every minterm of $\PM$ has width at least $2k$, and thus in any $\DNF$ computing $\PM$ every clause must be of width at least $2k$.
%Therefore, every clause in such a $\DNF$ evaluates to 1 on at most $2^{n+k}/2^{2k}=2^{n-k}$ inputs of the $\PM$ function.
%By \Cref{lem:preimages}, $|O| \geq \Omega\left(2^n\cdot\min(2^k,n-k+1)\right)$.
%Thus, every $\DNF$ for $\PM$ must contain at least $|O|/2^{n-k}=\Theta\left(2^k\cdot\min(2^k,n-k+1)\right)$ clauses.
%
%\medskip
%The proof idea for $\CNF$s is similar.
We say that a clause of a $\CNF$ \emph{covers an input $w \in \zo^{n+k}$} if this clause evaluates to 0 on $w$.
Note that a clause of width $c$ covers at most $2^{n+k-c}$ elements in $\zo^{n+k}$.
For the parameters $k \leq n$ we claim first that every clause of a $\CNF$ computing $\PM$
must be of width at least $c = c(k,n)$ depending on the range of $k$ (as follows from \Cref{lem:minterm}).
This implies that the number of clauses in any $\CNF$ computing $\PM$ is at least $|Z|/2^{n+k-c}$.
Below, we use \Cref{lem:preimages} and \Cref{lem:minterm} to estimate $c$ and $|Z|$ for different ranges of $k$.

If $k\leq\log{n}+1$, then $|Z|\geq \Omega\left(2^n(1-2^{-k})^n\right)$ by \Cref{lem:preimages}.
By \Cref{lem:minterm}, the width of each maxterm is at least $c\geq k+\frac{n-k+1}{k}$.
Thus, the number of clauses in any $\CNF$ computing $\PM$ must be at least
%Thus, every clause of the $\CNF$ can cover only $O(2^{n+k-(k+n/k)})$ zeros of $\PM$.
%Therefore, the number of clauses in each $\CNF$ is at least
%\[
%\Omega\left(|Z|/2^{n-n/k}\right)\geq \Omega\left(2^{n/k}(1-2^{-k})^n\right) \geq \Omega\left(2^{n/k} e^{n/(2^{k-1})}\right)=\Omega\left(2^{n(1/k-(2\log{e})/2^k)}\right) \; ,
%\]
%where the last inequality follows from $1-x/2 \geq e^{-x}$ for $x\in[0,1.5]$.
\[
\Omega\left(|Z|/2^{n+k-c}\right)\geq \Omega\left(|Z|/2^{n-n/k}\right)\geq \Omega\left(2^{n/k}(1-2^{-k})^n\right) \geq \Omega\left(2^{n/10k}\right) ,
\]
where the last bound follows from the inequality $2^{1/k} \cdot (1 - 2^{-k}) \geq 2^{1/10k}$ which holds for all $k \geq 2$.

For $k\geq \log{n} +1$, \Cref{lem:preimages} gives us an $\Omega(2^{n+k})$ lower bound on $|Z|$.
\Cref{lem:minterm} provides a lower bound on the width $c$ of maxterms: for $\log{n}+1\leq k\leq\sqrt{n}$, $c\geq k+n/k-1$, and for $k\geq\sqrt{n}$, $c\geq 2\sqrt{n-k+1}$. The desired bounds on the number of clauses in any $\CNF$ computing $\PM$ now follow immediately.
\end{proof}

\paragraph{Discussion.}
\Cref{lem:DNF} and \Cref{thm:CNF DNF} together give the lower bounds of \Cref{thm:depth2lb}. We observe a curious behavior of $\CNF$s and $\DNF$s for $\PM$. For $k\leq\sqrt{n}$, an optimal depth-$2$ circuit for $\PM$ is a $\DNF$. It can also be shown that for $k\geq n-O(\frac{n}{\log{n}})$, an optimal circuit is a $\CNF$. (Indeed, in order to certify that $\PM(x,y)=0$, it suffices to give mismatches for each of the $(n-k+1)$ shifts of the pattern $y$ in $x$. This amounts to $k^{O(n-k+1)}<n\cdot2^k$ clauses.) We leave the exact $\CNF$ complexity of $\PM$ for the regime $k>\sqrt{n}$ as an open problem. One way to prove a stronger lower bound in this regime would be to give a lower bound on the width of every maxterm. This approach does not lead to stronger lower bounds because there exist maxterms of width $2\sqrt{n}$. To see this, consider an assignment where the first $\sqrt{n}$ characters of the pattern $y$ are fixed to zeros, and all indices divisible by $\sqrt{n}$ in the text $x$ are fixed to ones. While we cannot prove a stronger lower bound on the width of ``most'' maxterms, we know that some maxterms must have width at least $n-k+1$. Indeed, consider the text $x=0^n$ and pattern $y=10^{k-1}$. Every clause which outputs $0$ on this pair, must assign the first $(n-k+1)$ positions of $x$ to $0$. 

We remark that weaker lower bounds of $2^{\Omega(\sqrt{n/k})}$ and $2^{0.08n/k}$ on the size of $\CNF$ computing $\PM$ follow from the reduction from Disjointness in \Cref{obs:disj} and the known lower bound on the depth-$3$ complexity of Iterated Disjointness~\cite{haastad1995top} and Disjointness~\cite{jukna2006graph}.

%%%%%%%%%%%%%%%%%%%%%%%%%%%%%%%%%%%%%%%%%%%%%%%%%%%%%%%%%%%%%
%\section{Lower bound for unbounded depth DeMorgan circuits}\label{sec:lb DeMorgan unbdd}
%\section{Proof of \Cref{thm:lb DeMorgan gate elimination}}\label{sec:lb DeMorgan unbdd}
\subsection{Lower bound for unbounded depth}\label{sec:unbounded}
%%%%%%%%%%%%%%%%%%%%%%%%%%%%%%%%%%%%%%%%%%%%%%%%%%%%%%%%%%%%%
%In this section we prove \Cref{thm:lb DeMorgan gate elimination}.
%\elimination*
Now we prove the lower bound of  \Cref{thm:lb DeMorgan gate elimination}. For circuits with fan-in $2$, a linear lower bound follows from the observation that $\PM$ essentially depends on all of its inputs. In the next lemma, we use an extension of the gate elimination technique to show that even in the class of DeMorgan circuits with \emph{unbounded fan-in}, $\PM$ still requires linear size.

\begin{lemma}
For $k>1$, any DeMorgan circuit computing $\PM$ has size at least $n/2$.
\end{lemma}
\begin{proof}
    Suppose that a circuit $C$ computes $\PM$, and consider an input $(x,y)$ to the circuit.
    We prove that $C$ has at least $n/2$ gates as follows. %using the gate elimination method.
    We show that for \emph{any} fixing of the bits $x_1,x_3,x_5,\ldots,x_{2t-1}$ for $1 \leq t \leq n/2-1$, the restricted function depends on the bit $x_{2t+1}$.
    Since the function depends on $x_{2t+1}$, any circuit computing it must have $x_{2t+1}$ or $\neg x_{2t+1}$ among its inputs. Without loss of generality we assume that $x_{2t+1}$ appears as an input. Now we show that we can fix the input $x_{2t+1}$ so that at least one gate of the circuit is removed.

        Indeed, if $x_{2t+1}$ appears as an input to an $\AND$ gate, we can set $x_{2t+1} = 0$, hence setting the output of the gate to $0$.
    This way we can remove the gate from the circuit by setting the output of the $\AND$ gate to $0$, and propagating it.
    (It is possible that we also affect other gates).
    Similarly, if $x_{2t+1}$ appears as an input to an $\OR$ gate, we can set $x_{2t+1} = 1$,
    hence setting the outputs of the $\OR$ gate to $1$, and remove the gate from the circuit.
    Therefore, we can remove at least $n/2-1$ gates from the circuit, and hence the size
    of the original circuit computing $\PM$ was at least $n/2$.

    Therefore, it is left to prove the following claim
    \begin{claim}
    Let $k \geq 2$.
        For \emph{any} fixing of the bits $x_1,x_3,x_5,\ldots,x_{2t-1}$ for $1 \leq t \leq n/2-1$, the restricted function depends on the bit $x_{2t+1}$.
    \end{claim}
    \begin{proof}
        Let $x^*=(x^*_1,x^*_3,\ldots,x^*_{2t-1})$ be the values of the $t$ fixed bits of $x$. In order to show that the restricted function depends on $x_{2t+1}$, we show that there exist two inputs: $(x,y)$ and $(x',y)$, such that $0=\PM(x,y)\neq\PM(x',y)=1$ and $(x,y)$ and $(x',y)$ are extensions of $x^*$ which differ only in the position $2t+1: x_{2t+1}\neq x'_{2t+1}$.

        We set all non-fixed bits of $x$ to $0$, except for $x_{2t+2}=1$.
        Now we set $x'$ to be equal to $x$ everywhere except for the position $2t+1$, where $x'_{2t+1}=1$.
        Now we see that the string $x$ does not contain two ones in a row, while $x'$ does. Since $k\geq2$, we can set $y$ to be an arbitrary substring of $x'$ of length $k$ which contains $x'_{2t+1}$ and $x'_{2t+2}$.
        By the definition of $y$ we have $\PM(x',y)=1$ and $\PM(x,y)=0$ because $x$ does not contain the substring $11$.
    \end{proof}
%\end{itemize}
\end{proof}
    %This completes the proof of \Cref{thm:lb DeMorgan gate elimination}.
%As already noted the $O(n \cdot \log^2{n})$
%and depth $O(\log^2{n})$ is mentioned in \cite{galil1985optimal} without a proof. Without this upper bound, our best upper bound on the size of an unbounded depth DeMorgan circuit computing $\PM$ is $O(nk)$ (see~\cref{thm:depth2lb}).

%%%%%%%%%%%%%%%%%%%%%%%%%%%%%%%%%%%%%%%%%%%%%%%%%%%%
%\section{Learning hidden patterns: Tight bounds on the $\VC$ dimension}
\section{Learning}
%%%%%%%%%%%%%%%%%%%%%%%%%%%%%%%%%%%%%%%%%%%%%%%%%%%%
\subsection{$\VC$ dimension}
In this section we prove Theorem~\ref{thm:VCmain}.
\VCmain*

%Let $\Sigma$ be a fixed finite alphabet of size $|\Sigma|\geq2$. By $\Sigma^n$ we denote the set of strings over $\Sigma$ of length $n$, and by $\Sigma^{\leq k}$ we denote the set of strings of length at most $k$. We study the $\VC$ dimension of the class of functions, where each function is identified with a pattern of length at most $k$, and outputs $1$ only on the strings containing this pattern. Recall that the length of the pattern $k=k(n)\leq n$ can be a function of $n$.

%For a fixed \emph{finite} alphabet $\Sigma$ and an integer $k>0$, let us define the class of functions $\mathcal{H}_{k,\Sigma}$ over $\Sigma^n$ as follows.
%Every function $h_{\sigma} \in \mathcal{H}_{k,\Sigma}$ is parameterized by a pattern $\sigma \in \Sigma^{\leq k}$ of length at most $k$. Hence,
%$|\mathcal{H}_{k,\Sigma}|=\frac{|\Sigma|^{k+1}-1}{|\Sigma|-1}$. For a string $s\in\Sigma^n$, $h_{\sigma}(s)=1$ if and only if $s$ contains $\sigma$ as a substring.
%In the special case of the binary alphabet $\Sigma=\{0,1\}$, we denote the hypotheses class $\mathcal{H}_{k,\Sigma}$ by $\mathcal{H}_k$.

%\subsection{Binary alphabet}
%We will show that $\VC(\H_{k, \Sigma})$ equals $\min(k\log{|\Sigma|}, \log{n})$ up to a small additive factor. First we prove this upper bound on the $\VC$ dimension of $\mathcal{H}_{k, \Sigma}$.
We begin by upper bounding the $\VC$ dimension.
In the proof we will use the following folklore construction of a Sperner system.
\begin{definition}
A system $\F$ of subsets of $\{1,\ldots,n\}$ is called a \emph{Sperner system} if no set in $\F$ contains another one:
$$\forall A, B \in \F\colon A\neq B \implies A\not\subseteq B\; .$$
\end{definition}
For any $n$, there exists a Sperner system of size ${n \choose \floor{n/2}}$.  Indeed, one can take $\F$ to be the family of all sets of size exactly $\floor{n/2}$.

\begin{lemma}\label{lem:upper}
Let $\Sigma$ be a finite alphabet of size $|\Sigma|\geq2$, then
$$\VC(\mathcal{H}_{k,  \Sigma}) \leq \min(\ceil{k\log{|\Sigma|}}, \log{n}+0.5\log\log{n}+2) \; .$$
\end{lemma}
\begin{proof}
Since $|\mathcal{H}_{k, \Sigma}|=\frac{|\Sigma|^{k+1}-1}{|\Sigma|-1}<2|\Sigma|^k$, $\H_{k, \Sigma}$ cannot shatter a set of
strings $S$ of size $|S|\geq k\log|\Sigma|+1$. Hence, $\VC(\mathcal{H}_{k, \Sigma})\leq \ceil{k\log|\Sigma|}$. We now give a different upper bound
on $\VC(\mathcal{H}_{k, \Sigma})$.

Suppose one can shatter some $d$ strings $X=\{x_1, \ldots, x_d\}$, where $x_i \in \Sigma^n$. That is, for any dichotomy of the strings from $X$, there is a pattern from $\Sigma^{\leq k}$ which realizes it. We will show an upper bound on $d=\VC(\mathcal{H}_{k, \Sigma})$.

Consider a Sperner system of size $D={d-1 \choose \floor{(d-1)/2}}$ of the set $\{1,\ldots,d-1\}$. Now add the element $d$ to each of these sets. This way we have $D$ sets containing the element $d$, such that none of them is a subset of another. Let us denote this family of $D$ sets by ${\cal S}=\{{\cal S}_1,\ldots,{\cal S}_D\}$. Consider the following set ${\cal D}=\{{\cal D}_1,\ldots,{\cal D}_D\}$ of $D$ dichotomies of $X$: ${\cal D}_i$ labels $x_j$ with one if and only if $j \in {\cal S}_i$.

From the assumption that $X$ can be shattered, we have that there exist $D$ patterns $p_1,\ldots,p_D$ which realize all $D$ dichotomies from $\cal D$. Since each of these dichotomies labels $x_d$ with one, the string $x_d$ must contain all patterns $p_i$. If one of the patterns $p_i$ was a substring of another pattern $p_j$, then we would have that ${\cal S}_i \subseteq {\cal S}_j$, which contradicts the definition of Sperner systems.

Thus, there must be $D$ patterns which are contained in the string $x_d$, and none of these patterns is a substring of another one. Let us sort the occurrences of these $D$ patterns in $x_d$ by their starting position. Since one pattern cannot be a substring of another one, their ending positions must form an increasing sequence. Therefore, the length $n$ of $x_d$ is at least $D$. This gives us that
$$\frac{2^{d-1}}{\sqrt{2(d-1)}} \leq {d-1 \choose \floor{(d-1)/2}} \leq D \leq n \; ,$$ or, $\VC(\mathcal{H}_{k,  \Sigma})\leq \log{n}+0.5\log\log{n}+2$.
%Consider the set $P$ of all patterns that shatter subsets of $\{x_1, \ldots, x_d\}$ which contain $x_1$. Since there are $2^{d-1}$ such subsets, there must be at least $2^{d-1}$ distinct patterns in $P$. The string $x_1$ must contain all of the patterns from $P$. Since $x_1\in\Sigma^n$ contains at most $nk$ distinct substrings of length from $1$ to $k$, we have that $2^{d-1} \leq nk \; .$
%For every $i \in [d]$ let $x_i(k)$ be the prefix of length $n-k+1$ of the string $x_i$. Clearly, for $i \in [d]$, $x_i(k)$ contains at least one character from a pattern in $L$. Furthermore, every pattern in $L$ appears by definition in at least $d/2$ such prefixes and there is a single pattern that can occur at a given position in a string $x_i$. It follows that
%$$\sum_{i=1}^d|x_i(k)|=(n-k+1)d \geq 2^{d-1}d/2. $$ This implies that $d \leq \log (n-k+1)+2$ concluding the proof.
%This implies that $\VC(\H_{k, \Sigma}) \leq \log{n}+\log{k}+1=O(\log{n})$.
\end{proof}

To lower bound the VC dimension of $\mathcal{H}_{k, \Sigma}$ we need the following lemma.
\begin{lemma}\label{lem:construction}
Let $m$ be an integer $m\geq1$, and $\Sigma$ be an alphabet of size $|\Sigma|\geq2$. There exists a set $T_m$ of at least $|\Sigma|^{m-1}$ strings from $\Sigma^{m+\ceil{\log m}+2}$ with the following property. For any two distinct strings $\tau_1,\tau_2 \in T_m$, their concatenation $\tau=\tau_1\circ\tau_2$ doesn't contain any string from $T_m\setminus\{\tau_1,\tau_2\}$ as a substring.
\end{lemma}
\begin{proof}
Since $|\Sigma|\geq2$, we can fix two distinct characters $0,1\in\Sigma$.
Let $S_m$ be the set of all strings from $\Sigma^{m}$ which don't contain $0^{\ceil{\log{m}}+1}$ as a substring. Note that each string containing $0^{\ceil{\log{m}}+1}$ as a substring is uniquely defined by the starting position of $0^{\ceil{\log{m}}+1}$ and $m-\ceil{\log{m}}-1$ remaining characters. Therefore, the number of strings containing $0^{\ceil{\log{m}}+1}$ doesn't exceed $m|\Sigma|^{m-\ceil{\log{m}}-1}$, and $|S_m|\geq |\Sigma|^m-m|\Sigma|^{m-\ceil{\log{m}}-1}\geq |\Sigma|^m-|\Sigma|^{m-1} \geq |\Sigma|^{m-1}$.

For each string $s \in S_m$, we include in $T_m$ the string $s$ appended with the string $0^{\ceil{\log m}+1}1$ ($\ceil{\log m}+1$ zeros followed by a one) in the end. Note that the number of strings in $T_m$ is at least $|\Sigma|^{m-1}$, and each string in this set is of length $m+\ceil{\log m}+2$. Now we'll prove that for any $\tau_1,\tau_2 \in T_m$, $\tau=\tau_1\circ\tau_2$ doesn't contain any string from $T_m\setminus\{\tau_1,\tau_2\}$.

Assume, towards contradiction, that $\tau$ contains a string $\tau_3 \in T_m\setminus\{\tau_1,\tau_2\}$. Recall that $\tau_1=s_1\circ0^{\ceil{\log m}+1}1, \tau_2=s_2\circ0^{\ceil{\log m}+1}1, \tau_3=s_3\circ0^{\ceil{\log m}+1}1$, where $s_1, s_2,$ and $s_3$ are distinct strings from $S_m$. Thus,
$$\tau=s_1\circ0^{\ceil{\log m}+1}1\circ s_2\circ0^{\ceil{\log m}+1}1 \; .$$
Since $\tau_3$ ends with $0^{\ceil{\log m}+1}1$ and neither $s_1\circ0^{\ceil{\log m}+1}$ nor $s_2\circ0^{\ceil{\log m}+1}$ contains this substring, $\tau_3$ must be equal to $\tau_1$ or $\tau_2$.
\end{proof}

\begin{lemma}\label{lem:VCLB}
Let $\Sigma$ be a finite alphabet of size $|\Sigma|\geq2$, then
$$\VC(\mathcal{H}_{k,  \Sigma}) \geq \min((k-\log{k}-5)\log{|\Sigma|}, \log{n}-\log\log{n}) \; .$$
\end{lemma}
\begin{proof}
We will show that there exist $d$ strings of length $n$: $x_0,\ldots, x_{d-1} \in \Sigma^n$, and $2^d$ patterns $p_0, \ldots, p_{2^d-1}\in\Sigma^{\leq k}$ of length at most $k$, such that each dichotomy of $\{x_0,\ldots, x_{d-1}\}$ is realized by some pattern $p_i$. This will prove that the $\VC$ dimension of $\H_{k, \Sigma}$ is at least $d$.

Let
$$m=\left\lfloor\min\left(k-\log{k}-3, \frac{\log{n}}{\log{|\Sigma|}}-\frac{\log\log{n}}{\log{|\Sigma|}}+1\right)\right\rfloor\; ,$$ and let $d=\floor{(m-1)\log{|\Sigma|}}$. If $m<1$, then the Lemma statement follows trivially, hence, assume that $m \geq 1$. We will show that $\H_{k, \Sigma}$ shatters $d$ strings, and this will finish the proof.

Let $T_m$ be the set of strings from Lemma~\ref{lem:construction} for the chosen value of $m$. Since the size $|T_m|\geq |\Sigma|^{m-1} \geq 2^d$, we can choose $2^d$ patterns from $T_m$. Let us call these patterns $p_0 \ldots p_{2^d-1}$. For $0\leq i \leq d-1$, we define $x_i$ to be the concatenation (in arbitrary order) of all strings $p_j$ such that the $i$th bit of the binary expansion of $j$ is $1$. If the length of $x_i$ is less than $n$, we pad it with ones in the end.

\begin{enumerate}
\item The length of each pattern $p_i$ is
$$m+\ceil{\log{m}}+2 \leq m+\log{m}+3 \leq k-\log{k}-3+\log{k}+3=k \; .$$
\item Each string $x_i$ can be padded to a string of length $n$, because it is a concatenation of $2^{d-1}$ patterns of total length
$$2^{d-1}(m+\ceil{\log{m}}+2) \leq 2^{(m-1)\log{|\Sigma|}-1}(m + \log{m}+3) \leq 2^{\log{n}-\log\log{n}-1}(\log{n}+4)\leq n  \; .$$
\item Consider now a subset $I \subseteq [d-1]$ of the strings $x_0,\ldots, x_{d-1}$ to be shattered. Let $0\leq j \leq 2^{d}-1$ be the number whose binary expansion is the indicator vector of $I$. We claim that the pattern $p_j$ realizes the set $I$. First, by the definition of the strings $x_i$, the pattern $p_j$ was among the patterns concatenated in $x_i$ if and only if $i\in I$. Second, by Lemma~\ref{lem:construction}, no $x_i$ with $i\not\in I$ contains $p_j$ as a substring.
\end{enumerate}
\end{proof}
This concludes the proof of Theorem~\ref{thm:VCmain}.

\subsection{Learning $\mathcal{H}_{k,  \Sigma}$}\label{subsec:learning}
In this section we discuss an efficient algorithm for learning the hypothesis class $\mathcal{H}_{k,  \Sigma}$. For completeness we state the definition of PAC learning:

Let $\mathcal{D}$ be a distribution over $\Sigma^n$. Suppose we are trying to learn $h_{\sigma}$
for $\sigma \in \Sigma^{\leq k}$. Given $\tau \in \Sigma^{\leq k}$, the loss of $h_{\tau}$ with respect to $h_{\sigma}$ is defined as
$$L_{\mathcal{D},\sigma}(\tau)=\Pr_{x \sim D}[h_{\tau}(x)\neq h_{\sigma}(x)]\; .$$

Following the notion of PAC-learning \cite{valiant1984theory,shalev2014understanding}, we can now define what we mean by learning $\mathcal{H}_{k,\Sigma}.$

\begin{definition}\label{def:PAC}
An algorithm $\cal A$ is said to PAC-learn $\mathcal{H}_{k,\Sigma}$ if for every distribution $\mathcal{D}$ over $\Sigma^n$ and
every $h_{\sigma} \in \mathcal{H}_{k,\Sigma}$ for all $\epsilon, \delta \in (0,1/2)$ the following holds.
Given $m:=m(\epsilon,\delta,n,k)$ i.i.d. samples $(x_1,h_\sigma(x_1)), \ldots, (x_m,h_\sigma(x_m))$ where each $x_i$
is sampled according to the distribution $\mathcal{D}$, $\cal A$ returns with probability at least $1-\delta$ a function $h_{\tau} \in \mathcal{H}_{k,\Sigma}$
such that $L_{\mathcal{D},\sigma}(\tau) \leq \epsilon$. Here the probability is taken with respect to the $m$ i.i.d. samples as well as the possible random choices
made by the algorithm $\cal A$.
\end{definition}

Throughout, we refer to $\delta$ as the confidence parameter and $\epsilon$ as the accuracy parameter.

In Definition~\ref{def:PAC} we consider the \emph{realizable} case.
Namely there exists $h_{\sigma} \in \mathcal{H}_{k,\Sigma}$ that we want to learn. One can also consider the \emph{agnostic} case. Consider a distribution $\mathcal{D}$ over $\Sigma^n \times \{0,1\}$. We now define the loss of $h_{\tau}$ as
$$L_{\mathcal{D}}(\tau)=\Pr_{x \sim D}[h_{\tau}(x)\neq y]\;,$$ namely the measure under $\mathcal{D}$ of all pairs $(x,y)\in \Sigma^n \times \{0,1\}$ with $h_{\tau}(x)\neq y$ \cite{shalev2014understanding}. In the agnostic case we wish to find, given $m$ i.i.d. samples $(x_1,h(x_1)), \ldots, (x_m,h(x_m))$, a pattern $\sigma' \in \Sigma^{\leq k}$ such that
$L_{\mathcal{D}}(\sigma') \leq \min_{\tau} L_{\mathcal{D}}(\tau)+\epsilon$ (where the minimum is taken over all $\tau \in \Sigma^{\leq k}$). Thus agnostically PAC-learning generalizes the realizable case where $\min_{\tau} L_{\mathcal{D}}(\tau)=0$.

Recall that a function $h_{\sigma} \in \mathcal{H}_{k,  \Sigma}$ (parameterized by the pattern $\sigma$ of length at most $k$)
can be learned with error $\epsilon$ and confidence $\delta$ by considering $m=O(\VC(\mathcal{H}_{k,  \Sigma}))$ samples $(x_1,h_\sigma(x_1)), \ldots, (x_m,h_\sigma(x_m))$ (where the constant in the $O$ term depends on $\epsilon,\delta$) and following the ERM (expected risk minimization) rule:
Finding $\sigma'$ that minimizes the loss
$$L(h_{\sigma'}):=\frac{|\{i \in [m]:h_{\sigma'}(x_i)\neq h_{\sigma}(x_i)\}|}{m} \; .$$ In words,
to PAC learn $h_{\sigma}$ we simply look for a string $\sigma'$ of length at most $k$ such that the fraction of sample points
that are misclassified by $h_{\sigma'}$ is minimized (the ERM rule applies both for the agnostic and realizable settings).

By Lemma~\ref{lem:upper}, the number of samples needed to PAC-learn $h_{\sigma}$ is at most
$O(\log n)$ (ignoring the dependency on $\epsilon,\delta$). Clearly we can implement the ERM by considering all possible substrings of length at most $k$ that occur in the
$m=O(\log n)$ strings
$x_1 \ldots x_m$
and finding the substring $\sigma'$ minimizing $L(h_{\sigma'})$. The number of such substrings is at most
$O(\log n\sum_{i=1}^k(n-k+1)) \leq O(kn \log n)$. Since for every substring we can check whether it occurs in a string of length $n$
in time $O(n)$, we can implement the ERM rule by going over every substring $\eta$ of length at most $k$ and checking for every string $x_i $ (with $i \in [m]$)
whether $\eta$ occurs in $x_i$. By keeping track of the pattern which has minimal classification error with respect to the sample $(x_1,h_\sigma(x_1)), \ldots, (x_m,h_\sigma(x_m))$
we can thus implement the ERM rule in time $O(kn^2 \log^2 n)$.

We can do better if the number of substrings of length at most $k$ which is upper bounded by $2|\Sigma|^{k}$ is smaller than $(n-k+1)\log{n}$.
Suppose for example, that $k \leq \frac{\log n}{\log|\Sigma|}$. By Lemma~\ref{lem:upper}, the VC-dimension of $\mathcal{H}_{k,  \Sigma}$
is then upper bounded by $k \log |\Sigma|$. Hence in this case we can assume the number of strings $m$ in our sample is at most
$k \log |\Sigma|$, and we can implement the ERM rule in time $O(|\Sigma|^{k} kn\log |\Sigma|)$. When $k,|\Sigma|$ are constants
independent of $n$ we can thus learn $h_{\sigma}$ in time $O(n)$.

We summarize this discussion with the following corollary:
\begin{corollary}\label{cor:runtime}
The hypothesis class $\mathcal{H}_{k,  \Sigma}$ is PAC-learnable in time $O(kn^2 \log^2 n)$, where the
$O$ symbol contains constants depending on $\epsilon,\delta$ but not on $n,k$.
If $k,|\Sigma|$ are constants independent of $n$, then $\mathcal{H}_{k,  \Sigma}$ can be learned in time $O(n)$.
\end{corollary}

\subsection{Extensions}\label{sec:extensions}

\paragraph{Infinite alphabet.}
So far we have been considering the case of finite alphabet $\Sigma$. For an infinite $\Sigma$ the $\VC$ dimension is essentially $\log{n}$ for every value of $k\geq 1$. Note that the upper bound of $\VC(\mathcal{H}_{k,  \Sigma}) \leq\log{n}+0.5\log\log{n}+2$ from Lemma~\ref{lem:upper} holds even for infinite alphabets $\Sigma$. Indeed, this upper bound counts the number of different patterns which have to occur in one string and compares it to the length of the string $n$. In the following lemma we give a lower bound of $\log{n}$ for all values of $k\geq1$.

\begin{lemma}
Let $\Sigma$ be an infinite alphabet, and $k\geq 1$. Then
$$ \VC(\mathcal{H}_{k,  \Sigma}) =(1+o(1)) \log{n} \; .$$
\end{lemma}
\begin{proof}
%Let us first prove the lower bound. Suppose one can shatter each subset of the $d$ strings $x_1, \ldots, x_d \in \Sigma^n$. Consider a Sperner system of size $D={d-1 \choose \floor{(d-1)/2}}$ of the set $\{1,\ldots,d-1\}$. Now add the element $d$ to each of these sets. This way we have $D$ sets containing the element $d$, such that none of them is a subset of another. Thus, there must be $D$ patterns $p_1,\ldots,p_D$ which are contained in the string $x_d$. Note that the property of Sperner system guarantees that none of these patterns can be a substring of another one. Let us sort the occurrences of these $D$ patterns in $x_d$ by their starting position. Since one pattern can't be a substring of the other one, their ending positions must form an increasing sequence. Therefore, the length $n$ of $x_d$ is at least $D$. This gives us that $2^{d-1}/\sqrt{2(d-1)} \leq D \leq n$, or, $\VC(\mathcal{H}_{k,  \Sigma})\leq \log{n}+0.5\log\log{n}+2$.

For the lower bound, we pick $2^d+1$ distinct elements $\bot,a_0,\ldots, a_{2^d-1}\in\Sigma$. Let  $d=\floor{\log{n}}+1$. We construct $d$ strings $X=\{x_1,\ldots,x_d\}, x_i\in\Sigma^n$ such that any dichotomy of $X$ is realizable in $\Sigma^{\leq k}$. Now, for $0\leq i\leq d-1$, we define $x_i$ to be a concatenation (in arbitrary order) of all $a_j$ such that the $i$th bit of the binary expansion of the number $j$ is $1$. Note that now the length of each $x_i$ is at most $2^{d-1}\leq n$, so we pad each $x_i$ with the element $\bot$ so that $x_i\in \Sigma^n$.

Now we need to show that each dichotomy of $X$ is realizable. For a dichotomy $\cal D$ of $X$, consider the set $I\subseteq[d]$ such that $\cal D$ labels $x_i$ with one if and only if $i \in I$.
In order to realize $\cal D$, we take the pattern $a_j$ such that the binary expansion of $0\leq j\leq 2^d-1$ equals the indicator vector of $I$. By the definition of $x_i$, $x_i$ contains $a_j$ if and only if $i \in I$. Therefore, $a_j$ realizes the dichotomy $\cal D$. Note that the pattern $a_j\in\Sigma\subseteq\Sigma^{\leq k}$.
\end{proof}

\paragraph{Learning multiple patterns.}
In this section we make a few simple observations regarding the VC dimension of classifiers defined by the occurrences of multiple patterns.
 The main observation is that learning a \emph{constant} number of patterns does not change the asymptotics of the $\VC$ dimension so long as the number of patterns is upper bounded by the
 length of the pattern $k$. Let us consider two natural classes $\mathcal{H}_{k,  \Sigma}^{\text{and}}$ and $\mathcal{H}_{k,  \Sigma}^{\text{or}}$ of multi-pattern Boolean functions over $\Sigma^n$. Each function $h_\sigma^{\text{and}}\in\mathcal{H}_{k,  \Sigma}^{\text{and}}$ is parameterized by $c>0$ patterns $\sigma=(\sigma_1,\ldots,\sigma_c)\in\left(\Sigma^{\leq k}\right)^c$. Now, for an $s\in\Sigma^n$, $h_{\sigma}^{\text{and}}(s)=1$ if and only if $s$ contains \emph{each} $\sigma_i,1\leq i\leq c$ as a substring (for brevity we omit from notation the dependence of $\mathcal{H}_{k,  \Sigma}^{\text{and}}$ and $\mathcal{H}_{k,  \Sigma}^{\text{or}}$ on $c$). Similarly, a function $h_\sigma^{\text{or}}\in\mathcal{H}_{k,  \Sigma}^{\text{or}}$ takes the value one: $h_{\sigma}^{\text{or}}(s)=1$ if and only if $s$ contains \emph{at least one} $\sigma_i$ as a substring. We stress that we assume that the set of patterns $\sigma_i,i \in [c]$ are distinct.

An upper bound on the VC dimension of $\mathcal{H}_{k,  \Sigma}^{\text{and}}$ and $\mathcal{H}_{k,  \Sigma}^{\text{or}}$ follows at once from the following Lemma proved in
\cite{blumer1989learnability} (Lemma 3.2.3).
\begin{lemma}
Let $\mathcal{H}_1,\ldots,\mathcal{H}_c$ be classes of functions of $\VC$ dimension at most $\forall i\colon \VC(\mathcal{H}_i)\leq d$. Let
\begin{align*}
\mathcal{H}^{\text{and}}&=\{f_{h_1,\ldots,h_c}(x)=h_1(x)\wedge\ldots\wedge h_c(x)\colon h_1\in\mathcal{H}_1,\ldots,h_c\in\mathcal{H}_c\} \; , \\
\mathcal{H}^{\text{or}}&=\{f_{h_1,\ldots,h_c}(x)=h_1(x)\vee\ldots\vee h_c(x)\colon h_1\in\mathcal{H}_1,\ldots,h_c\in\mathcal{H}_c\} \; .
\end{align*}
Then $\VC(\mathcal{H}^{\text{and}})=O(dc\log{c})$ and $\VC(\mathcal{H}^{\text{or}})=O(dc\log{c})$.
\end{lemma}

%We show that for any constant $c$,
%\begin{align*}
%\VC(\mathcal{H}_{k,  \Sigma}^{\text{and}}), \VC(\mathcal{H}_{k,  \Sigma}^{\text{or}}) = \Theta(\min\left(\log{|\Sigma|}(k-O(\log k)), \log n+O(\log \log n)\right)) \; .
%\end{align*}
We now turn to the lower bound. Our result here is rather modest: We show that the lower bound on the VC dimension of a single pattern also holds for $\mathcal{H}_{k,  \Sigma}^{\text{and}}$ and $\mathcal{H}_{k,  \Sigma}^{\text{or}}$ provided that the number $c$ of (distinct) patterns is not too large. Let us see that the lower bounds of Lemma~\ref{lem:VCLB} hold for $\mathcal{H}_{k,  \Sigma}^{\text{and}}$ and $\mathcal{H}_{k,  \Sigma}^{\text{or}}$. Indeed, for the class $\mathcal{H}_{k,  \Sigma}^{\text{and}}$, we use the construction from Lemma~\ref{lem:VCLB}, where for every pattern $\sigma$ in that construction we consider a set of $k$ patterns $\{\sigma^{1},\ldots,\sigma^{k}\}$. We define $\sigma^i=\sigma_1\ldots\sigma_i$ to be the prefix of length $i$ of $\sigma$. For example, for the pattern $11010$ we  take the patterns $\{1,11,110,1101,11010\}$. We remark that we obtain $k$ distinct subpatterns of $\sigma$. Since every string from the shattered set contains $\sigma$ if and only if it contains every pattern from $\{\sigma^1, \ldots ,\sigma^k\}$, all dichotomies are realized by the ``last'' pattern $\sigma^k=\sigma$. Since $c\leq k$, we take $c$ longest patterns $\{\sigma^{k-c+1}, \ldots ,\sigma^k\}$, and our construction gives a shattered set of size %$\min\left(\log{|\Sigma|}(k-O(\log k)), \log n+O(\log \log n)\right)$
%If $k<c=O(1)$, then for a fixed alphabet $\Sigma$ the desired lower bound is constant, and thus holds trivially. Therefore, we have:
\begin{align*}
\VC(\mathcal{H}_{k,  \Sigma}^{\text{and}}) \geq\min\left(\log{|\Sigma|}(k-O(\log k)), \log n+O(\log \log n)\right) \; .
\end{align*}

For the class $\mathcal{H}_{k,  \Sigma}^{\text{or}}$, we can take $T'_m\subseteq T_m$ with $|T'_m|=|T_m|/2$ and shatter a set of size $d-1$. Now for every $\sigma \in T'_m$ define a $c$-tuple of patterns by adding to $\sigma$ $c-1$ patterns in $T_m \setminus T'_m$ (where $c \leq 2^{d-1}-1$ because $c\leq k$). % Again, if $c\geq2^{d-1}$, then $k,\Sigma=O(1)$, and the lower bound holds trivially}
%Since all strings in  which never appears in our construction (for example, $0^{2\log{m}+2}$, or a string from $T_m$ which is not used in the construction of the shattered set).
Since none of the strings in the shattered set contains a pattern from $T_m \setminus T'_m$, all dichotomies are realized by the ``first'' pattern $\sigma_1$. Again, our construction from Lemma~\ref{lem:VCLB} gives a shattered set of size $\min\left(\log{|\Sigma|}(k-O(\log k)), \log n+O(\log \log n)\right)-1$.

To conclude, we have proved:
\begin{theorem}
%\begin{align*}
Let $1\leq c\leq k$ be a fixed constant. Then
$$VC(\mathcal{H}_{k,  \Sigma}^{\text{and}}),VC(\mathcal{H}_{k,  \Sigma}^{\text{or}}) = \Theta\left(\min\left(\log{|\Sigma|}(k-O(\log k)), \log n+O(\log \log n)\right)\right) \; .$$
%\end{align*}
\end{theorem}
%\begin{lemma}[folklore]
%Let $\mathcal{H}_1,\ldots,\mathcal{H}_c$ be classes of functions of $\VC$ dimension at most $\forall i\colon \VC(\mathcal{H}_i)\leq d$. Let
%\begin{align*}
%\mathcal{H}^{\text{and}}&=\{f_{h_1,\ldots,h_c}(x)=h_1(x)\wedge\ldots\wedge h_c(x)\colon h_1\in\mathcal{H}_1,\ldots,h_c\in\mathcal{H}_c\} \; , \\
%\mathcal{H}^{\text{or}}&=\{f_{h_1,\ldots,h_c}(x)=h_1(x)\vee\ldots\vee h_c(x)\colon h_1\in\mathcal{H}_1,\ldots,h_c\in\mathcal{H}_c\} \; .
%\end{align*}
%Then $\VC(\mathcal{H}^{\text{and}})=O(dc\log{c})$ and $\VC(\mathcal{H}^{\text{or}})=O(dc\log{c})$.
%\end{lemma}
%\begin{proof}
%For a class of Boolean functions $\mathcal{H}$, let $\Pi_{\mathcal{H}}(m)$ be its growth function:
%\begin{align*}
%\Pi_{\mathcal{H}}(m)=\max_{x_1,\ldots,x_m}\left|\{(h(x_1),\ldots,h(x_m))\colon h\in\mathcal{H}\}\right|\; .
%\end{align*}
%From Sauer's lemma~\cite{vapnik1971uniform,sauer1972density}, for any $m\geq d$,
%\begin{align*}
%\forall i\colon \Pi_{\mathcal{H}_i}(m)\leq \left(\frac{em}{d}\right)^d \; .
%\end{align*}
%Since $\Pi_{\mathcal{H^{\text{and}}}}(m) \leq \prod_{i=1}^c \Pi_{\mathcal{H}_i}(m)\leq\left(\frac{em}{d}\right)^{dc}$, we have that for $m=O(dc\log{c}), %\Pi_{\mathcal{H^{\text{and}}}}(m)<2^m$, and $\VC(\mathcal{H^{\text{and}}})<m$. Similarly, $\VC(\mathcal{H^{\text{or}}})<m$.
%\end{proof}

\paragraph{Patterns of length $k$.}
One can also consider learning patterns of length \emph{exactly} $k$. We consider this case separately since it seems that getting tight bounds on $\VC$-dimension in this case is a harder task. In particular, we are not able to get tight bounds for the regime $k=n^{1-o(1)}$ and leave this as an open question.

For a fixed \emph{finite} alphabet $\Sigma$ and an integer $k>0$, the class of functions $\mathcal{E}_{k,\Sigma}$ over $\Sigma^n$ is defined as follows.
Every Boolean function $h_{\sigma} \in \mathcal{E}_{k,\Sigma}$ is parameterized by a pattern $\sigma \in \Sigma^{k}$ of length exactly $k$. Therefore,
$|\mathcal{E}_{k,\Sigma}|=|\Sigma|^{k}$. For a string $s\in\Sigma^n$, $h_{\sigma}(s)=1$ if and only if $s$ contains $\sigma$ as a substring.
We use a simple double counting argument to prove:
\begin{lemma}\label{thm:VCeq}
$\VC(\mathcal{E}_{k,  \Sigma}) \leq \min(k \log |\Sigma|, \log (n-k+1)+1)$.
\end{lemma}
\begin{proof}
Since $|\mathcal{E}_{k,  \Sigma}|=|\Sigma|^k$, the upper bound of $\VC(\mathcal{E}_{k,  \Sigma}) \leq k \log |\Sigma|$ follows immediately. For the other upper bound, suppose we can shatter a set $X$ of $d$ strings $x_1 \ldots x_d$ of length $n$. Then we have $2^d$ distinct patterns which realize all dichotomies of $X$.
For a fixed $i \in [d]$, the number of dichotomies of $X$ which label $x_i$ with one is $2^{d-1}$. Therefore, every string $x_i$ contains at least $2^{d-1}$ distinct patterns of length $k$.
On the other hand, any string of length $n$ can contain at most $n-k+1$ distinct patterns of length $k$. Thus, we have $2^{d-1} \leq n-k+1$, or,
$d \leq \log (n-k+1)+1$.
\end{proof}

Now we prove the following upper bound:
\begin{lemma}
Let $\Sigma$ be a finite alphabet of size $|\Sigma|\geq2$, then
$$\VC(\mathcal{E}_{k,  \Sigma}) \geq \min((k-\log{k}-5)\log{|\Sigma|}, \log{n}-\log{k}) \; .$$
\end{lemma}
\begin{proof}
Let
\begin{align*}
d&=\floor{ \min((k-\log{k}-5)\log{|\Sigma|}+1 , \log{n}-\log{k}+1)} \; ,\\
m&=\floor{k-\log{k}-2} \; .
\end{align*}

 By Lemma~\ref{lem:construction}, we have the set $T_m$ of $|\Sigma|^{m-1} \geq 2^{(k-\log{k}-5)\log{|\Sigma|}}\geq2^d$ strings of length $k$. We choose $2^d$ arbitrary strings $p_0,\ldots,p_{2^d-1}$ from $T_m$. Now we essentially use the construction from Lemma~\ref{lem:VCLB}: we construct $d$ strings $x_0,\ldots,x_{d-1} \in \Sigma^n$ such that $x_i$ contains $p_j$ if and only if the $i$th bit of the binary expansion of $j$ is $1$, and pad $x_i$ with ones to have $x_i\in\Sigma^n$.

 We have $d$ strings which can be shattered by $\mathcal{E}_{k,  \Sigma}$. We also know that the length of each pattern is $k$. We only need to show that before the padding step, each string $x_i$ had length at most $n$. Since each $x_i$ is a concatenation of $2^{d-1}$ patterns of length $k$, we have that its length is at most
$$2^{d-1}\cdot k\leq 2^{\log{n}-\log{k}}\cdot k=n \; .$$
\end{proof}

We remark that for the case of patterns of length \emph{at most $k$}, Lemma~\ref{lem:upper} and Lemma~\ref{lem:VCLB} give essentially tight bounds for all regimes of the parameters. Here, in the case of patterns of length \emph{exactly $k$}, we have a gap between lower and upper bounds for the regime $k=n^{1-o(1)}$.

\paragraph{2D patterns.}
Our bounds for learning one dimensional strings generalize to the 2D case. Here we have an $n \times n$ image over an alphabet
$\Sigma$ and am $m \times m$ pattern $\sigma$ where $m \leq k \leq n$. An image is classified as $1$ if and only if it contains $\sigma$.
%We now wish to learn the pattern $\sigma$ from i.i.d. samples taken of these classifier. Formally:
\begin{definition}
For a fixed \emph{finite} alphabet $\Sigma$ and an integer $k>0$, let us define the class of Boolean functions $\mathcal{G}_{k,\Sigma}$ over $\Sigma^{n\times n}$ as follows.
Every function $g_{\sigma} \in \mathcal{G}_{k,\Sigma}$ is parameterized by a square 2D pattern $\sigma \in \Sigma^{m \times m}$ of dimension $m \leq k$. For a 2D image $s\in\Sigma^{n\times n}$ of dimension $n$, $g_{\sigma}(s)=1$ if and only if $s$ contains $\sigma$ as a consecutive sub-matrix (sub-image).
\end{definition}

We give tight bounds (up to low order terms) on $\VC(\mathcal{G}_{k, \Sigma})$. Since the proofs are very similar to the 1D case, we only sketch the arguments here.

Since $|\mathcal{G}_{k, \Sigma}|=\sum_{1\leq i \leq k}|\Sigma|^{i^2}+1\leq \sum_{1\leq i \leq k}|\Sigma|^{ik}+1<2|\Sigma|^{k^2}$, we have that $\VC(\mathcal{G}_{k, \Sigma})\leq \ceil{k^2\log |\Sigma|}$. Suppose that $\mathcal{G}_{k, \Sigma}$ shatters a set of $d$ 2D images from $\Sigma^{n\times n}$. By considering a
Sperner system over $\{1, \ldots, d-1\}$ of size $D={d-1 \choose \floor{(d-1)/2}}$ and adding the element $d$ to
each subset, we get a family of $D={d-1 \choose \floor{(d-1)/2}}$ patterns all lying in a single $n \times n$ image such that no pattern contains another one. We have that the bottom right corners of all these
patterns are distinct, and thus $\frac{2^{d-1}}{\sqrt{2(d-1)}} \leq D \leq n^2$ implying that $d \leq 2\log n+0.5\log\log{n}+3$. Hence,
$$\VC(\mathcal{G}_{k, \Sigma})\leq \min(\ceil{k^2\log|\Sigma|},2\log n+0.5\log \log n+3).$$

For the lower bound, the main observation is that we can generalize Lemma~\ref{lem:construction} to the two dimensional case
having a set $R_m$ of $(m+ 2\lceil\log m\rceil +2)\times (m +2\lceil\log m\rceil +2)$ 2D patterns of cardinality $|\Sigma|^{m^2-1}$ such
that for any four distinct patterns $\alpha_1,\alpha_2, \alpha_3,\alpha_4$ from $R_m$, their concatenation (fitting the four patterns into a $2(m+ 2\lceil\log m\rceil +2)\times 2(m +2\lceil\log m\rceil +2)$ square image in each of the $4!$ possible ways) does not contain any $\alpha_5 \neq \alpha_{i}$ for $1\leq i\leq 4$ from $R_m$.
%By concatenation we mean putting two templates right next to each other (creating a rectangular pattern of dimension $2(m+ \lceil\log m\rceil +2)\times (m +\lceil\log m\rceil +2)$)
%or right above each other (creating a rectangular pattern of dimension $(m+ \lceil\log m\rceil +2)\times 2(m +\lceil\log m\rceil +2)$).
We achieve this by taking all $m\times m$ templates not containing the all $0$ 2D square template of size $(2\lceil\log m\rceil+1)\times(2\lceil\log m\rceil+1)$, padding them by an all zero strip of width $2\lceil\log m\rceil+1$ on the right and bottom, and then adding a boundary of ones on those two sides. Similarly to Lemma~\ref{lem:construction}, it can be verified that $R_m$ satisfies the desired condition.

We now set
$$m=\left\lfloor\min\left(k-2\log{k}-4, \sqrt{\frac{2\log{n}}{\log{|\Sigma|}}-\frac{3\log\log{n}}{\log{|\Sigma|}}}\right)\right\rfloor\;.$$
Let $R_m$ be a set of $|\Sigma|^{m^2-1}$ templates whose construction was described in the paragraph above and set $d=\lfloor(m^2-1)\log |\Sigma|\rfloor$.
Since $|R_m|= |\Sigma|^{m^2-1}\geq 2^d$, we can choose $2^d$ distinct 2D patterns $q_0 \ldots q_{2^{d}-1}$ from $R_m$.
The dimension of each pattern $q_i$ is $m+2\ceil{\log m}+2$ which by the choice of $m$ is at most $k$.

Define a set of $n \times n$ images $Y:=\{y_0 \ldots y_{d-1}\}$ where $y_i$ is an image containing all the patterns $q_j$ from $R_m$
such that the binary expansion of $j$ equals $1$ in the $i$th location.
%More precisely, we take all patterns $q_j$ from $R_m$
%such that the binary expansion of $j$ equals $1$ in the $i$th position, and pack them in adjacent horizontal strips of %dimension $n \prod m+2\ceil{\log m}+2$ in $y_i$ where in every such strip the leftmost
%pattern touches the left boundary of $y_i$.
This way, each image from $Y$ must contain at most $2^{d-1}$ patterns, while we can fit $\left\lfloor\frac{n}{m+2\ceil{\log{m}}+2}\right\rfloor^2$ patterns into an image of size $n\times n$.
It can be verified that for the chosen values of $m$ and $d$, $2^{d-1} \leq \left\lfloor\frac{n}{m+2\ceil{\log{m}}+2}\right\rfloor^2$. Thus, we have that each $y_i$ can be padded to an $n \times n$ image if necessary by assigning $1$ to all unassigned positions. Finally, it follows in a similar fashion to the 1D case that
the set of patterns $q_0 \ldots q_{2^{d}-1}$ shatters $Y$. Hence $R_m$ shatters $Y$. Since $|Y|=d$ the VC dimension of the set of all 2D patterns of dimensions at most $k$ is at least $d$.

We conclude this discussion with the following Theorem:
\begin{theorem}
$$\VC(\mathcal{G}_{k, \Sigma})= \min\left((k-O(\log k))^2\log|\Sigma|,2\log n-O(\log \log n)\right) \; .$$
\end{theorem}

\section*{Acknowledgements}
We thank Pawe{\l} Gawrychowski for his useful feedback and Gy. Tur\'{a}n for sharing \cite{groeger1993linear} with us. We are also very grateful to anonymous reviewers for their insightful comments.
%%%%%%%%%%%%%%%%%%%%
\bibliographystyle{alpha}
\bibliography{pm}

\end{document}